\documentclass[11pt]{article} %***
\usepackage[sectionbib]{natbib}
\usepackage{array,epsfig,rotating}
\usepackage{bm}
\usepackage[]{hyperref}  
\usepackage{sectsty, secdot}
\sectionfont{\fontsize{12}{14pt plus.8pt minus .6pt}\selectfont}
\subsectionfont{\fontsize{12}{14pt plus.8pt minus .6pt}\selectfont}
\usepackage{amsmath}
\usepackage{amssymb}
\usepackage{arydshln}
\usepackage{amsfonts}
\usepackage{multirow}
\usepackage{amsthm}
\usepackage{makecell}
\usepackage{cases}
\usepackage{epstopdf}
\usepackage{float}
\usepackage{physics}
\usepackage{enumerate}
\usepackage{enumitem}
\usepackage[ruled,linesnumbered]{algorithm2e}
\usepackage{tikz}  
\usetikzlibrary{arrows.meta}
\newtheorem{theorem}{Theorem}[section]

\newtheorem{lemma}{Lemma}[section]
\newtheorem{corollary}{Corollary}[section]

\theoremstyle{definition}

\newtheorem{example}{Example}
\theoremstyle{remark}
\newtheorem{remark}{Remark}[section]
\def\T{{ \mathrm{\scriptscriptstyle T} }}
\def\I{{ \textbf{I} }}

\def\t{{ \textbf{t} }}

\def\tr{{ \text{tr} }}

\usepackage[title]{appendix}
\usepackage{titlesec}
\usepackage{etoolbox, chngcntr}
\AtBeginEnvironment{appendices}{%
 \titleformat{\section}{\bfseries}{\appendixname~\thesection:}{0.3em}{}%
\counterwithin{equation}{section}
}

\newcommand{\blind}{0}

    \makeatletter
    \renewcommand\section{\@startsection {section}{1}{\z@}%
                                       {-3.5ex \@plus -1ex \@minus -.2ex}%
                                       {2.3ex \@plus.2ex}%
                                       {\normalfont\fontfamily{phv}\fontsize{16}{19}\bfseries}}
    \renewcommand\subsection{\@startsection{subsection}{2}{\z@}%
                                         {-3.25ex\@plus -1ex \@minus -.2ex}%
                                         {1.5ex \@plus .2ex}%
                                         {\normalfont\fontfamily{phv}\fontsize{14}{17}\bfseries}}
    \renewcommand\subsubsection{\@startsection{subsubsection}{3}{\z@}%
                                        {-3.25ex\@plus -1ex \@minus -.2ex}%
                                         {1.5ex \@plus .2ex}%
                                         {\normalfont\normalsize\fontfamily{phv}\fontsize{14}{17}\selectfont}}
    \makeatother
\begin{document}
\def\spacingset#1{\renewcommand{\baselinestretch}%
{#1}\small\normalsize} \spacingset{1}

\if0\blind
{
  \title{\bf Optimal Designs for Gamma Degradation Tests}
    \author{Hung-Ping Tung$^{a,*}$, Yu-Wen Chen$^a$ \\
$^a$ Department of Industrial Engineering and Management,\\
National Yang Ming Chiao Tung University, Hsinchu, 30010, Taiwan, ROC\\
*Corresponding author: hptung@nycu.edu.tw}
\date{}
  \maketitle
} \fi

\if1\blind
{
  \bigskip
  \bigskip
  \bigskip
  \begin{center}
    {\LARGE\bf Optimal Designs for Gamma Degradation Tests}
\end{center}
  \medskip
} \fi

\bigskip
\begin{abstract}
This paper analytically investigates the optimal design of gamma degradation tests, including the number of test units, the number of inspections, and inspection times. We first derive optimal designs with periodic inspection times under various scenarios. Unlike previous studies that typically rely on numerical methods or fix certain design parameters, our approach provides an analytical framework to determine optimal designs. In addition, the results are directly applicable to destructive degradation tests when number of inspection is one. The investigation is then extended to designs with aperiodic inspection times, a topic that has not been thoroughly explored in the existing literature. Interestingly, we show that designs with periodic inspection times are the least efficient. We then derive the optimal aperiodic inspection times and the corresponding optimal designs under two cost constraints. Finally, two examples are presented to validate the proposed methods and demonstrate their efficiency in improving reliability estimation.
%illustrate the superiority of the optimal design with the optimal inspection interval over the equal inspection interval.
\end{abstract}

\noindent%

\def\thefigure{\arabic{figure}}
\def\thetable{\arabic{table}}
\fontsize{11}{14pt plus.8pt minus .6pt}\selectfont

\spacingset{1.5}
% \if0\blind
% { \clearpage
% } \fi

\section{Introduction}
In an era of rapid technological evolution and intensifying market competition, quality, cost, and delivery have become critical concerns for both customers and manufacturers. Effectively assessing product reliability and accurately characterizing the associated uncertainties are essential for informed decision-making in areas such as scheduling, maintenance, and spare parts inventory \citep{aven2024fifty}. A common approach for assessing product reliability is through degradation tests (DTs), which monitor quality characteristics (QCs) that deteriorate over time \citep{meeker2022statistical}. This approach enables the estimation of product lifetime without waiting for actual failures. Among various statistical models for analyzing degradation data, the gamma process is widely adopted due to its monotonic nature, which aligns well with many physical degradation mechanisms. Numerous studies have applied the gamma process in reliability and operations research \citep{lawless2004covariates, park2005accelerated, park2006stochastic, hsu2008capability, hao2015led, balakrishnan2017gamma, wang2021degradation, palayangoda2022evaluation, li2022general, bautista2022condition, aven2024fifty}.

Given the value of degradation data in reliability analysis, obtaining high-quality data is critical, especially when testing resources are limited. Consequently, it is important to design efficient DTs that balance information gain with resource constraints. Key decision variables include the number of test units, the number of inspections and inspection times. The optimization of such test plans based on the gamma process has been extensively studied in the literature. 
% \ref{tab:reference} summarizes the papers related to DT, constant stress accelerated degradation test (CSADT), step stress accelerated degradation test (SSADT) and what decision variables they determine. Cost column denote the paper include cost constraint or not. The column ``$n$" denote the number of test units, where ``$\Delta$" means only determine the proportion and did not determine the total number of test units. The column ``$m$" denote the number of inspections. The column ``$t$" denote the inspection time. The ``$\Delta$" denote the solution is obtain under equal inspection interval. The column ``$x$" denote the stress levels which only happen in ADT. The ``$\Delta$" denote they did not determine the number of stress levels. From the table, we can first find that the discussion for the DT is lacked of a well study to investigate how parameter ans decision variable effect the planning. Second, all literature address the effect of only for equal inspection interval, which are convenient and practical for real-world applications. However, this may overlook the critical role of inspection time planning, potentially affecting the efficiency of the designs. In addition, only use numerical method.  Third, Tweedie process is a general stochastic process that include gamma process. However, most related literature use saddle point approximation to simplify the computation of Fisher information matrix, which replcae the gamma function in gamma process and loss some information of the second derivative of gamma function which is trigamma function.
Table~\ref{tab:reference} summarizes the literature of planning on DTs, constant-stress accelerated degradation tests (CSADTs), and step-stress accelerated degradation tests (SSADTs), along with the decision variables considered. The ``Cost'' column denotes whether cost constraints are included. The column ``$n$'' denotes the number of test units. The column ``$m$'' denotes the number of inspections. The column ``$t$'' denotes the inspection times, where ``$\Delta$'' indicates that periodic inspection times are assumed. The column ``$x$'' represents stress levels, where ``$\Delta$'' indicates that the number of stress levels is not determined.

% From the table, we observe that most planning problems are focus on the accelerated tests and the research of DT lacks a comprehensive study. Second, all reviewed studies assume equal inspection intervals. While this is convenient and practical for real-world applications, it may overlook the importance of optimizing inspection timing, which could impact the efficiency of the test design. Additionally, only numerical methods apply to determine the inspection time. All the analytical method did not determine the inspection time. Third, several studies work on the Tweedie process which is a general stochastic process that includes the gamma process as a special case. However, most related works adopt saddlepoint approximations to simplify the computation of the Fisher information matrix. This approach sacrifices part of the information in the second derivative of gamma function—the trigamma function.

From the table, most studies focus on planning ADTs, while research on DTs lacks a comprehensive investigation. Second, several studies consider the Tweedie process, a general stochastic process that includes the gamma process as a special case. However, most of these works employ saddlepoint approximations to simplify the computation of the objective function. This approximation sacrifices information contained in the second derivative of the logarithm of gamma function, known as the trigamma function or the polygamma function with order one. Third, all literature assume periodic inspection times. Although this assumption is convenient and practical for real-world applications, it may overlook the potential to improve testing efficiency by optimizing inspection times. Furthermore, only numerical methods have been used to determine inspection times, and none of the analytical approaches address this aspect. It is worth noting that inspection time planning is not a concern in two other well-known DT models, the Wiener process \citep{lim2011optimal, hu2015optimum, zhao2021accelerated} and the inverse Gaussian process \citep{peng2015inverse, fang2022inverse, peng2022optimum}, since the structure of their objective functions depends only on the termination time and is independent of inspection time planning.

\begin{table}[ht!]
    \centering
    \caption{Summary of references related to planning gamma degradation tests}
    \resizebox{\columnwidth}{!}{
    \begin{tabular}{|c|c|c|c|c|c|c|c|c|}
        \hline
        Reference  & model & criterion & cost & $n$&$m$&$t$&$x$& method \\
        \hline
        \cite{tsai2012optimal} &\makecell{DT\\[-10pt]random effect}& V & O & O & O & $\Delta$ & & numerical\\
        \hline
        \cite{lim2015optimum} & \makecell{CSADT\\[-10pt]one-variable}& V & O & O & O & $\Delta$ & $\Delta$& numerical\\
        \hline
        \cite{zhang2015reliability}& \makecell{CSADT\\[-10pt]one-variable}& V & O & O & O & $\Delta$ & & numerical\\
        \hline
        \cite{duan2019optimal} &\makecell{CSADT\\[-10pt]one-variable \\[-10pt]random effect}& D, A, V&  & O & & & O& analytical\\
        \hline
        \cite{tseng2016optimum} & \makecell{CSADT\\[-10pt]one-variable \\[-10pt] Tweedie process}& V&  & O & & & $\Delta$& analytical\\
        \hline
        \cite{lee2020global} & \makecell{CSADT\\[-10pt]one-variable \\[-10pt] Tweedie process}& V& O & O & O & $\Delta$ & $\Delta$& semi-analytical \\
        \hline
        \cite{tung2022analytical} & \makecell{CSADT\\[-10pt]one-variable \\[-10pt] Tweedie process}& V& O & O & O &  & O& analytical \\
        \hline
        \cite{tsai2015optimal} & \makecell{CSADT\\[-10pt]two-variable}& V & O & O &  & $\Delta$ & &numerical \\
        \hline
        \cite{tung2024optimizing}& \makecell{CSADT\\[-10pt]two-variable}& D, A, V & O & O & O &  & O & semi-analytical\\
        \hline
        \cite{limon2020designing}& \makecell{CSADT\\[-10pt]multi-variable}& V & O & O & O & $\Delta$ & $\Delta$& numerical\\
        \hline
        \cite{tseng2009optimal} & SSADT & V & O & O & O & $\Delta$ & & numerical \\
        \hline
        \cite{pan2014optimal} & \makecell{SSADT\\[-10pt]two QCs} & V & O & O & O & $\Delta$ & & numerical \\
        \hline
        \cite{yan2023optimal} & \makecell{SSADT\\[-10pt]Tweedie process} & D, V &  &  & O &  & O & analytical \\
        \hline
        This study & DT & D, A, V & O & O & O & O & & analytical \\
        \hline
    \end{tabular}    
    }
    \label{tab:reference}
\end{table}

As a result, this work aims to provide a comprehensive study of fundamental DTs from an analytical perspective, addressing a critical gap in the understanding of gamma DTs. The analytical results not only offer practical insights for real-world applications but also serve as a foundation for extending to more complex models, such as those involving random effects and accelerated tests. Specifically, the contributions of this study are as follows:
\begin{enumerate}
    \item The optimal designs under periodic inspection time (referred to as type-I) are analytically investigated across several scenarios. Additionally, several properties of the polygamma function are established.
    \item The optimal type-I designs under cost constraints are analytically derived, and a systematic method is proposed that can be readily implemented by practitioners.
    \item The impact of inspection time is thoroughly investigated, and the optimal designs under aperiodic inspection time (referred to as type-II) are analytically derived.
    \item To enhance practical applications, this study incorporates a minimum inspection interval constraint, which is a natural consideration in real-world applications but has not been discussed in existing literature.
\end{enumerate}

The remaining part of this study is organized as follows. Section 2 introduces the statistical inference of gamma DTs, and defines the optimal criteria of $D$-optimality, $A$-optimality, and $V$-optimality.  Section 3 presents the optimal type-I designs. Section 4 presents the optimal type-II designs. Two examples are provided in Section 5 to illustrate the efficiency of the proposed optimal designs. Finally, the conclusion is presented in Section 6.

\section{Problem formulation} \label{section: assumption}
\subsection{Model Assumptions for Gamma degradation tests}
Suppose we have a degradation test with $n$ test units. The QC of each unit is inspected $m$ times at time $t_j$ for $j=1,\ldots,m$, and $t_m=T$ is the termination time. Let $\Delta t_j = t_j - t_{j-1}$ denote the inspection interval with $t_0=0$.
In this study, we further impose a minimum inspection interval constraint, requiring $\tau \geq \Delta t$ for some $\Delta t>0$. This constraint is necessary because $\tau = 0$ is impractical in real-world applications, as inspection tools always record quality characteristics at discrete time points. The selection of $\Delta t$ can be determined using two approaches. The first approach is based on the minimum measurable interval of the inspection tool. However, if $\Delta t$ is too small, the inspected quality characteristics at times $t$ and $t+\Delta t$ may be identical due to the tool's resolution, leading to $z(t) = z(t+\Delta t)$.   
To address this issue, the second approach determines $\Delta t$ such that
    \begin{equation*}
    P(Z(\Delta t) > a) = b,
    \end{equation*}
where $a$ is a given threshold, typically relate to the tool's resolution, and $b$ is a specified probability. The values of $a$ and $b$ can be determined based on the practitioner's requirements.

In this notation, a degradation test is defined by $\zeta = (n, m, T, \t)$, where $\t=(\Delta t_1,\ldots,\Delta t_m)$. Let $\{Z(t) | t \geq 0\}$ denote a gamma process. For $i=1,\ldots,n$, let $Z_{ij}=Z_i(t_j)$ denotes the inspected QC of the $i$th unit at time $t_j$, then $Z_{ij}$ follows a gamma distribution with shape parameter $\alpha t_j$ and rate parameter $\beta$, denoted as Gam$(\alpha t_j, \beta)$. Specifically, $Z_{ij}$ has the following properties:
\begin{enumerate}
\item At the starting time $t_0$, the quality characteristic $Z_{i0}=0$.
\item Independent increments: 
For any increasing time sequence $0 =t_0 < t_1 < t_2 < \cdots < t_m < \infty$, the increments $\Delta Z_{i1}, \Delta Z_{i2}, \ldots, \Delta Z_{im}$, where $\Delta Z_{ij}= Z_{ij}-Z_{i(j-1)}$, are mutually independent.
\item Stationary increments: The distribution of $Z_i(t_l) - Z_i(t_k)$ is identical to that of $Z_i(t_l - t_k)$ for any time points $t_k < t_l$.
\item Each increment $\Delta Z_{ij}$ follows a gamma distribution with probability density function (PDF)
\begin{equation}
f(\Delta z_{ij};\Delta t_j)= \frac{\beta^{\alpha \Delta t_{j}}}{\Gamma(\alpha \Delta t_{j})}\Delta z_{ij}^{\alpha \Delta t_{j} -1}e^{-\beta \Delta z_{ij}}.
\label{eq:con_gamma_pdf}
\end{equation}
where $\Gamma(x)=\int_{0}^{\infty}s^{x-1}e^{-s}ds$ is the gamma function.
\end{enumerate}

While equation (\ref{eq:con_gamma_pdf}) offers a standard parametrization for gamma processes, this study adopts a Tweedie parametrization by denoting $\gamma=\log\alpha/\beta$ and $\alpha=\alpha$. The corresponding PDF is 
\begin{equation*}
f(\Delta z_{ij};\Delta t_j)= \frac{\Delta z_{ij}^{\alpha\Delta t_{j}-1}\alpha^{\alpha\Delta t_{j}}}{\Gamma(\alpha\Delta t_{j})}e^{\alpha(-\Delta z_{ij}e^{-\gamma}-\gamma\Delta t_{j})}
\label{eq:tweedie_gamma_pdf}
\end{equation*}
As mentioned in \cite{tung2024optimizing}, the Tweedie parametrization offers two advantages. First, when applying maximum likelihood estimation (MLE) to estimate the parameters, the Tweedie parametrization leads to asymptotic independence between $\hat{\alpha}$ and $\hat{\gamma}$, simplifying the computation. Second, the expectation of $Z_{ij}$ is $e^\gamma t_j$, which forms a log-linear model. In gamma ADTs, log-linear models such as the Arrhenius or generalized Eyring models are commonly used. The Tweedie parametrization helps interpret the physical or chemical meaning of $\gamma$ and enables the extension of our results to ADTs.

\subsection{Statistical inference and optimal criteria}
In this study, the MLE is used to estimate the parameter $\bm{\theta}=(\alpha,\gamma)^\T$. According to the properties of gamma process, the log-likelihood function for a degradation test can be written as 
\begin{equation}\label{eq:loglikelihoodfunction}
l(\bm{\theta})=\sum\limits_{i=1}^{n}\sum\limits_{j=1}^{m}-\log\Gamma(\alpha \Delta t_{j})+(\alpha \Delta t_{j}-1)\log \Delta z_{ij}+\alpha\Delta t_{j}\log\alpha+\alpha (-\Delta z_{ij}e^{-\gamma}-\gamma\Delta t_j).
\end{equation}
Then, we can obtain the MLE $\hat{\bm{\theta}}$ by maximizing the above equation. In addition, the Fisher information matrix is
\begin{equation*}\label{eq:fisherinformationmatrix}
    \I(\bm{\theta};\zeta)=E\left(-\frac{\partial^2 l(\bm{\theta})}{\partial \bm{\theta} \partial \bm{\theta}^\T}\right)=n\left[
    \begin{matrix}
    \sum\limits_{j=1}^m[\Delta t_j^2\psi_1(\alpha\Delta t_j)]-\frac{T}{\alpha}&0\\
    0&\alpha T
    \end{matrix}
    \right],
\end{equation*}
where $\psi_r(x)=d^{r+1}\log\Gamma(x)/dx^{r+1}$ is the polygamma function with order $r$. The asymptotic covariance matrix of $\hat{\bm{\theta}}$ is the inverse of $\I(\bm{\theta};\zeta)$.

The purpose of this study is to determine a $\zeta$ such that we can efficiently estimate the reliability information of the product. In this study, we consider three criteria: $D$-optimality, $A$-optimality and $V$-optimality.
\begin{enumerate}
    \item The $D$-optimality criterion aims to minimize the determinant of the approximate covariance matrix of estimates. The corresponding objective function is
        \begin{equation*}            
        \phi_D(\zeta)=\det(\I(\bm{\theta};\zeta)^{-1})
        =\frac{1}{n^2T\left\lbrace\alpha \sum\limits_{j=1}^m\left[\Delta t_j^2\psi_1(\alpha\Delta t_j)\right]-T\right\rbrace},
            \label{eq:D_obj}
        \end{equation*}
        which does not depend on $\gamma$.
    \item The $A$-optimality criterion aims to minimize the approximate variance of the estimates. This is achieved by minimizing the trace of the covariance matrix. The corresponding objective function is
        \begin{equation*}
        \phi_A(\zeta)=\tr(\I(\bm{\theta};\zeta)^{-1})
        =\frac{1}{n} \left\lbrace\frac{1}{\sum\limits_{j=1}^m\left[\Delta t_j^2\psi_1(\alpha\Delta t_j)\right]-\frac{T}{\alpha}}+\frac{1}{\alpha T}\right\rbrace,
            \label{eq:A_obj}
        \end{equation*}
        which does not depend on $\gamma$.
    \item The $V$-optimality criterion aims to minimize the approximate variance of the estimated value of the $p$-th quantile of product lifetime. Let $Q$ denote the lifetime of a product, which is defined as the time that the degradation path $\{Z_t \mid t \geq 0\}$ first time touches a threshold $\eta$. More specifically,
    \begin{equation*}
    Q = \inf \{ t \mid Z_t \geq \eta \}.
    \end{equation*}
    Let $F_Q(t; \theta)$ be the cumulative distribution function (cdf). Based on the monotonic property of gamma process, 
    \begin{equation*}
        F_Q(t;\theta)=P(Q\leq t)=P(Z_t\geq\eta)=\int_\eta^\infty f(z;t)dz.
    \end{equation*}
    Let $\xi_p = F_Q^{-1}(p; \bm{\theta})$ denote the $p$-th quantile of $Q$. By applying $\delta$-method, the corresponding objective function is
        \begin{equation*}
        \phi_V(\zeta)=\text{Avar}(\hat{\xi}_p)=\textbf{h}^\T \I(\bm{\theta};\zeta)^{-1}\textbf{h}
            =\frac{1}{n} \left\lbrace\frac{h_1^2}{\sum\limits_{j=1}^m\left[\Delta t_j^2\psi_1(\alpha\Delta t_j)\right]-\frac{T}{\alpha}}+\frac{h_2^2}{\alpha T} \right\rbrace,
            \label{eq:V_obj}
        \end{equation*}
        where
        \begin{equation*}
            \textbf{h}=(h_1,h_2)^\T=\frac{\partial \xi_p}{\partial \bm{\theta}}=\frac{-1}{f_Q(\xi_p;\bm\theta)}\frac{\partial F_Q(\xi_p;\bm\theta)}{\partial\bm\theta}, 
        \end{equation*}
        and $f_Q$ is the pdf of $Q$.
\end{enumerate}
The $D$- and $A$-optimality focus on the efficiency of the parameters estimation, which suitable be used for the practitioner who wants to understand the mechanism of the product. On the other hand, the $V$-optimality focuses on the efficiency of product lifetime estimation, which is mainly used in the degradation test.

%Note that, without affecting the results, we may treat $n$ and $m$ as continuous variables in certain proofs.

\section{Optimal type-I designs}
Let $\tau = T/m$ denote the equal inspection interval. In this case, the decision variables include $\zeta=(n,m,\tau)$, and the optimization problem is 
\begin{equation*}
    \begin{aligned}
    \min\limits_{\zeta=(n,m,\tau)} \quad &\phi(\zeta),\\
    \text{s.t.}\quad\quad&\tau\geq\Delta t,\\
    &n,m\geq 1,
    \end{aligned}
\end{equation*}
where $\phi_D$, $\phi_A$ and $\phi_V$ are redefined as
\begin{equation*}
\phi_D(\zeta) = \frac{\alpha^2}{(nm)^2 (\alpha^3 \tau^3 \psi_1(\alpha \tau) - \alpha^2\tau^2)},
\end{equation*}
\begin{equation*}
\phi_A(\zeta) = \frac{1}{nm} \left[\frac{\alpha^2}{\alpha^2\tau^2 \psi_1(\alpha \tau) - \alpha\tau} + \frac{1}{\alpha \tau} \right],
\end{equation*}
\begin{equation*}
\phi_V(\zeta)
= \frac{1}{nm} \left[\frac{\alpha^2 h_1^2}{\alpha^2\tau^2 \psi_1(\alpha \tau) - \alpha\tau} + \frac{h_2^2}{\alpha \tau} \right].
\end{equation*} 

% Obviously, the objective function decreases as $n$ and $m$ increase. However, in practice, limited resources prevent the use of an infinite number of test units or inspections. In the following, we consider three scenarios that reflect practical constraints commonly encountered in real applications. The first scenario is $n$ and $m$ are fixed. This scenario reflect the number of test units may been restrict because of the chamber slots, and inspections have been determined for example The destructive degradation test is belong to this scenario with $m=1$. In this scenario, the only need to determine is the inspection time. The second scenario is $n$ and $m\tau$ are fixed that is the termination time is fixed. This scenario happens when the reliability need to be obtained in a specific time. The third scenario is the most general that the $n$, $m$ and $\tau$ are restrict under a cost constraint.

Obviously, the objective function decreases as $n$ and $m$ increase. However, in practice, limited resources prevent the use of an infinite number of test units or inspections. In the following, we consider three scenarios that reflect practical constraints commonly encountered in real applications. The first scenario assumes that both $n$ and $m$ are fixed. This situation reflects cases where the number of test units is limited by factors such as chamber capacity, and the number of inspections is predetermined. For example, destructive degradation tests fall into this category with $m = 1$. In this scenario, the only decision variable to determine is the inspection interval. The second scenario assumes that $n$ and the total test duration $T=m\tau$ are fixed, meaning the experiment must be completed within a specific time frame. This situation arises in practice when reliability information is needed by a strict deadline, such as for a product launch, or scheduled maintenance planning.
The third scenario is the most general, where $n$, $m$, and $\tau$ are subject to a cost constraint. This setup reflects realistic planning situations where budget limitations affect all aspects of the test design.

\subsection{Optimal designs when $n$ and $m$ are fixed}
In this section, the only decision variable is $\tau$. The following two theorems summarize the optimal $\tau$ for three criteria.
\begin{theorem}\label{thm:D-optnocost}
    For D-optimality, $\phi_D(\zeta)$ decreases as $\tau$ increases.
\end{theorem}
\begin{proof}
    Let $x=\alpha \tau>0$. To show that $x^3\psi_1(x)-x^2$ is an increasing function, we begin by differentiating it:
    \begin{equation*}
        3x^2\psi_1(x)+x^3\psi_2(x)-2x.
    \end{equation*}
    This expression can be written in integral form as
    \begin{equation}
        x^3\int_0^\infty e^{-xs} \varpi(s)ds,\label{eq:d-opt-proof}
    \end{equation}
    where
    \begin{equation*}
        \varpi(s)=3\int_0^s \frac{z}{1-e^{-z}}dz-\frac{s^2}{1-e^{-s}}-2s.
    \end{equation*}
    Since $\lim_{s\rightarrow 0^+} \varpi(s)=0$ and $e^s(1-e^{-s})^2\varpi'(s)>0$, it follows that $\varpi(s)>0$ for $s\in\mathbb{R}^+$. Therefore, the integral in equation (\ref{eq:d-opt-proof}) is positive, and hence $x^3\psi_1(x)-x^2$ is an increasing function.
    %Since $\lim_{s\rightarrow 0^+} \varpi(s)=\lim_{s\rightarrow 0^+}e^s(1-e^{-s})^2\varpi'(s)=\lim_{s\rightarrow 0^+}(e^s(1-e^{-s})^2\varpi'(s))'=0$ and $(e^s(1-e^{-s})^2\varpi'(s))''=se^{s}+2-2e^{-s}>0$,  it follows that $\varpi(s)>0$ for $s\in\mathbb{R}^+$. Therefore, the integral in equation (\ref{eq:d-opt-proof}) is positive, and hence $x^3\psi_1(x)-x^2$ is an increasing function.
\end{proof}

\begin{theorem} \label{thm:V-optnocost}
    For V-optimality, if $$\frac{h_2^2}{\alpha^2h_1^2}\geq\frac{2}{3},$$ then $\phi_V(\zeta)$ decreases as $\tau$ increases. Otherwise, it attains a unique minimum at $$\tau =\max\left\lbrace \Delta t,~ \frac{1}{\alpha}\Omega^{-1}\left(-\frac{h_2^2}{\alpha^2h_1^2}\right)\right\rbrace.$$
\end{theorem}

\begin{proof}
    Let $x=\alpha\tau$. Differentiating the objective function yields
    \begin{equation}
        \left[\frac{\alpha^2 h_1^2}{x^2 \psi_1(x) - x} + \frac{h_2^2}{x} \right]'=-\frac{\alpha^2 h_1^2}{x^2}\left(\Omega(x)+\frac{h_2^2}{\alpha^2h_1^2}\right).\label{eq:V-opt_derivative}
    \end{equation}
    According to Lemma \ref{lemma:tung2025}
    \begin{equation*}
        \left[\frac{\alpha^2 h_1^2}{x^2 \psi_1(x) - x} + \frac{h_2^2}{x} \right]'<0, \forall x>0\text{ if and only if }\frac{h_2^2}{\alpha^2h_1^2} \geq \frac{2}{3},
    \end{equation*}
    and
    \begin{equation*}
        \left[\frac{\alpha^2 h_1^2}{x^2 \psi_1(x) - x} + \frac{h_2^2}{x} \right]'=0\text{ has a unique root at }x=\Omega^{-1}\left(-\frac{h_2^2}{\alpha^2h_1^2}\right)\text{ if and only if }\frac{h_2^2}{\alpha^2h_1^2} < \frac{2}{3}.
    \end{equation*}
\end{proof}

\begin{corollary}
    Theorem \ref{thm:V-optnocost} reduces to A-optimality when $h_1^2/h_2^2=1$.
\end{corollary}

In general, intuition suggests that using more resources provides more information. This holds true for D-optimality. However, for A- and V-optimality, the optimal inspection interval is capped under certain parameter settings, contradicting this intuition. Both the existence and nonexistence of the optimal $\tau$ have been observed in several applications. For example, the LED data in \cite{lim2015optimum} demonstrates the existence of an optimal $\tau$. The parameter estimates are $\alpha = 0.065$ and $\gamma = -0.77$, with $\eta = 0.5$ and $p = 0.1$, yielding $\frac{h_2^2}{\alpha^2 h_1^2} = 0.53 < \frac{2}{3}$. Thus, the optimal $\tau$ exists and is $53.2$ hours.
In contrast, for the carbon-film resistor data in \cite{tseng2009optimal}, the parameter estimates are $\alpha = 2.26 \times 10^{-4}$ and $\gamma = -11.12$, with $\eta = 5$ and $p = 0.05$, yielding $\frac{h_2^2}{\alpha^2 h_1^2} = 122.87 > \frac{2}{3}$. %118.26
Therefore, the optimal $\tau$ does not exist.
Fig. \ref{fig:optIVnocost} plots $\phi_V(\zeta)$ for the two examples under $n = m = 1$, verifying our theorem. The y-axis is presented on a log-scale.
\begin{figure}
    \centering
    \includegraphics[width=0.4\textwidth]{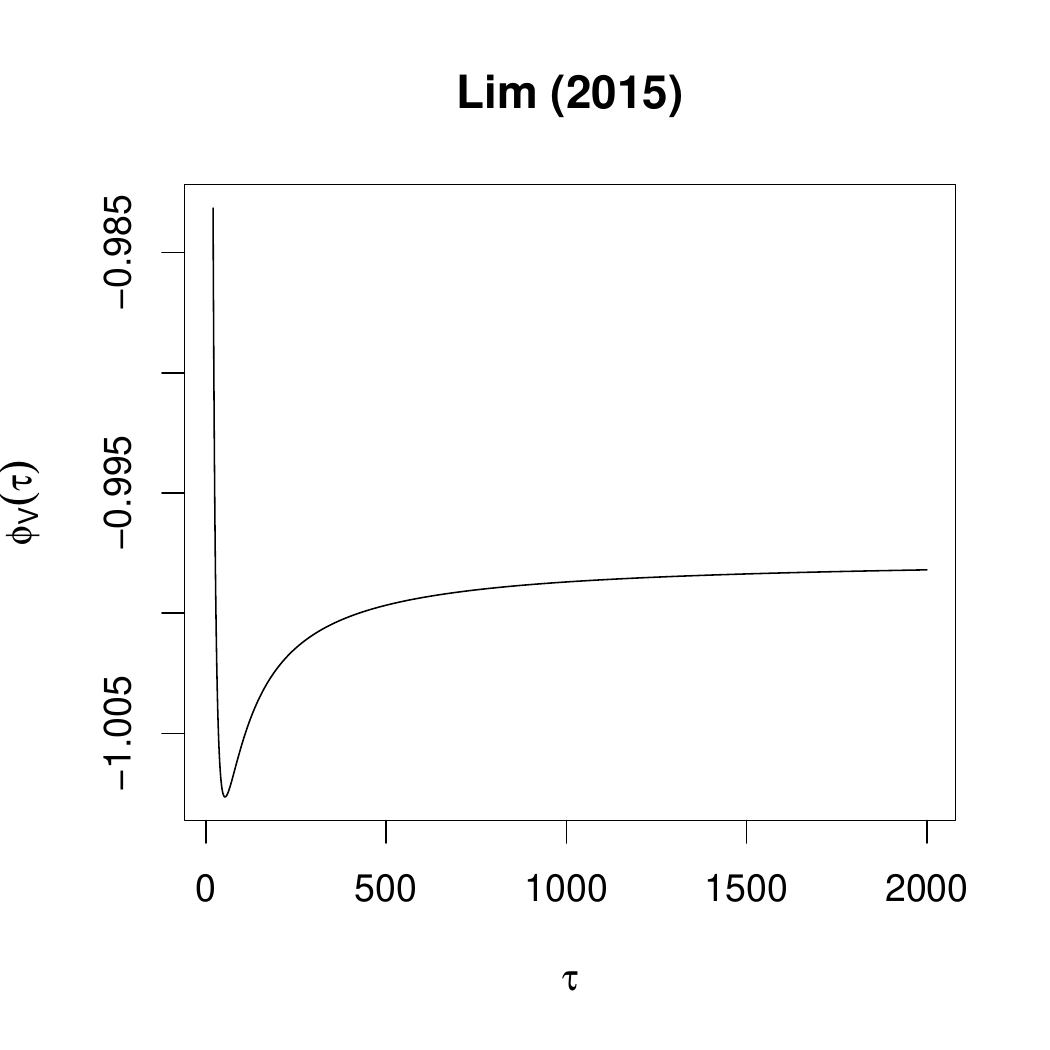}
    \includegraphics[width=0.4\textwidth]{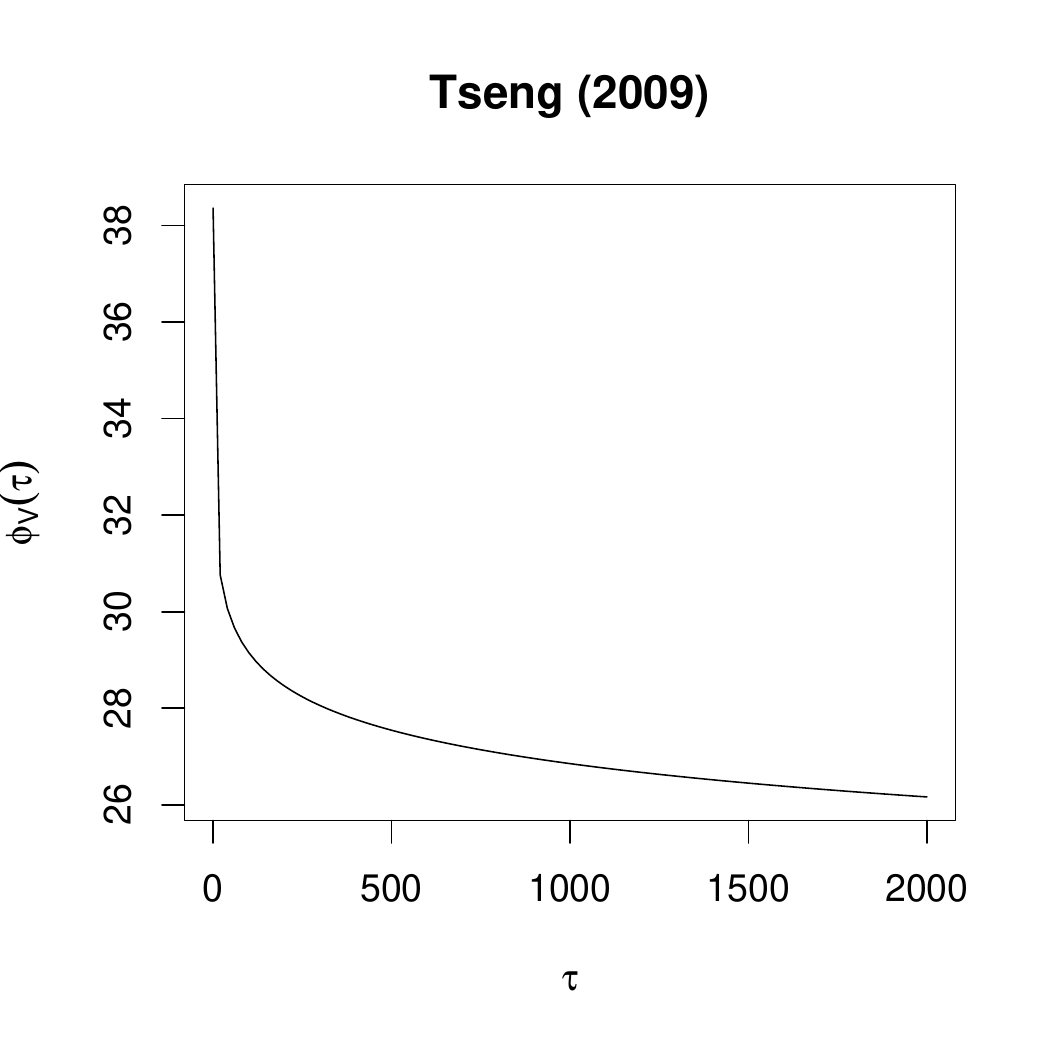}
    \caption{The log-scaled $\phi_V(\zeta)$ for the examples from \cite{lim2015optimum} and \cite{tseng2009optimal}.}
    \label{fig:optIVnocost}
\end{figure}

To provide practical insight into the index $h_2^2/(\alpha h_1)^2$, Fig. \ref{fig:Vcont} presents contour plots of $\log(h_2^2/(\alpha h_1)^2)$ with respect to $\alpha$ and $\gamma$ under $\eta = 1$ for $p = 0.05$, $0.3$, and $0.5$.
\begin{figure}
    \centering
    \includegraphics[width=0.32\textwidth]{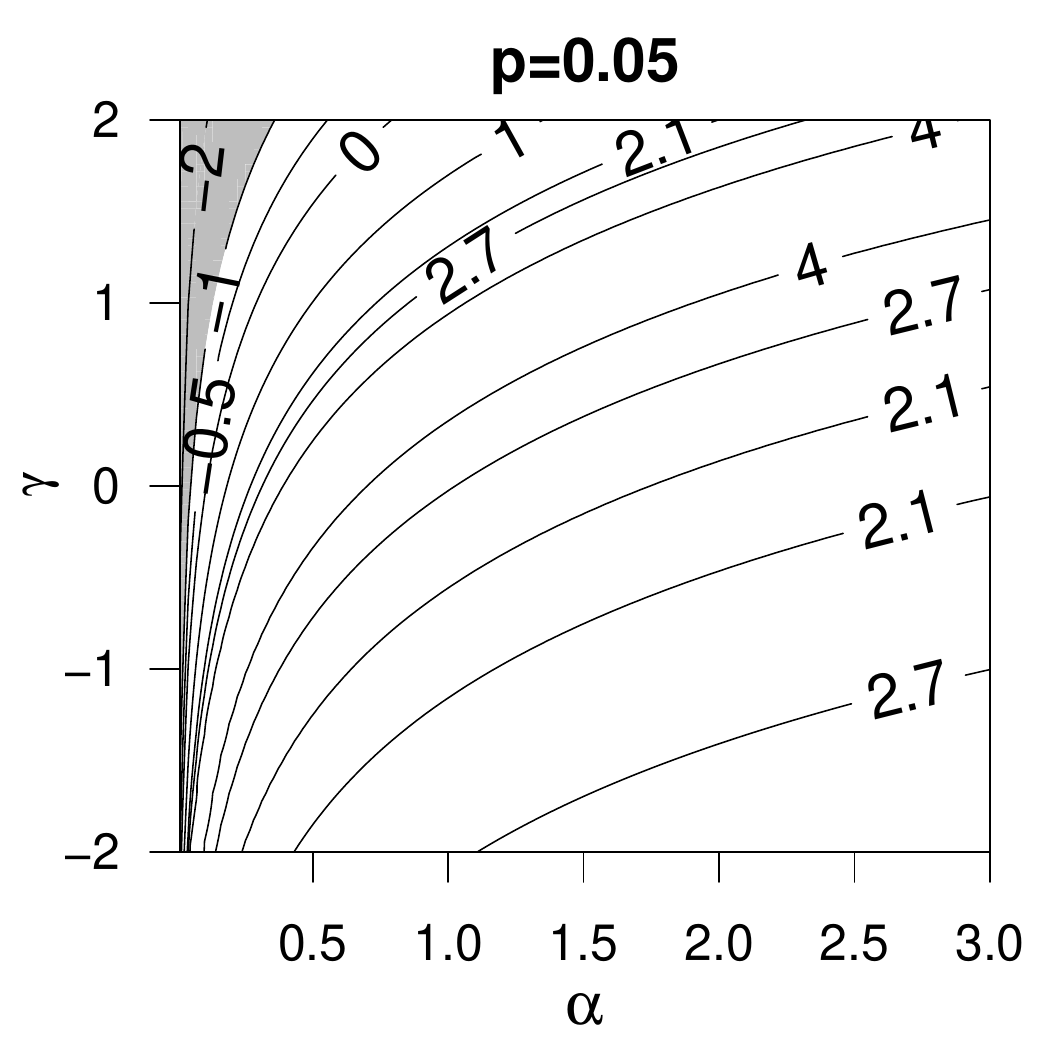}
    \includegraphics[width=0.32\textwidth]{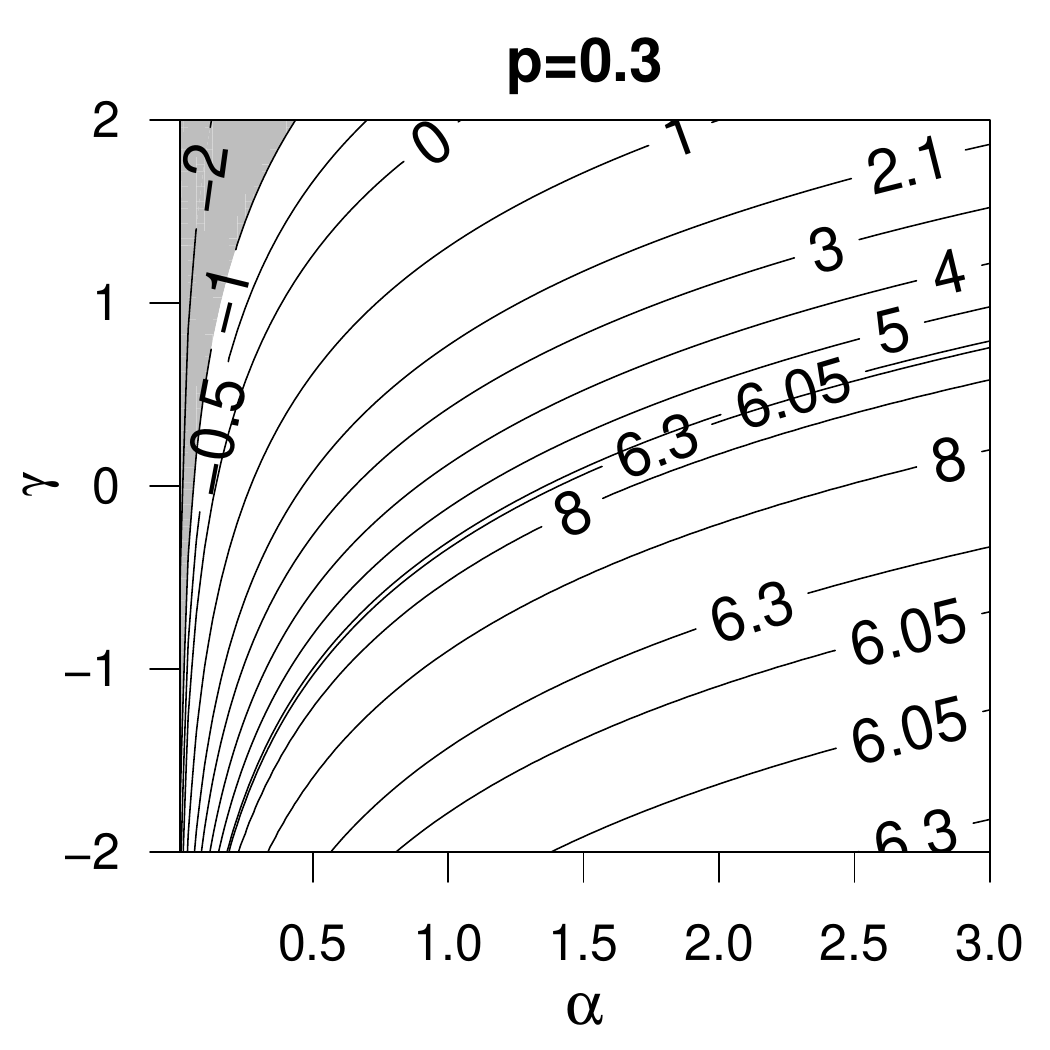}
    \includegraphics[width=0.32\textwidth]{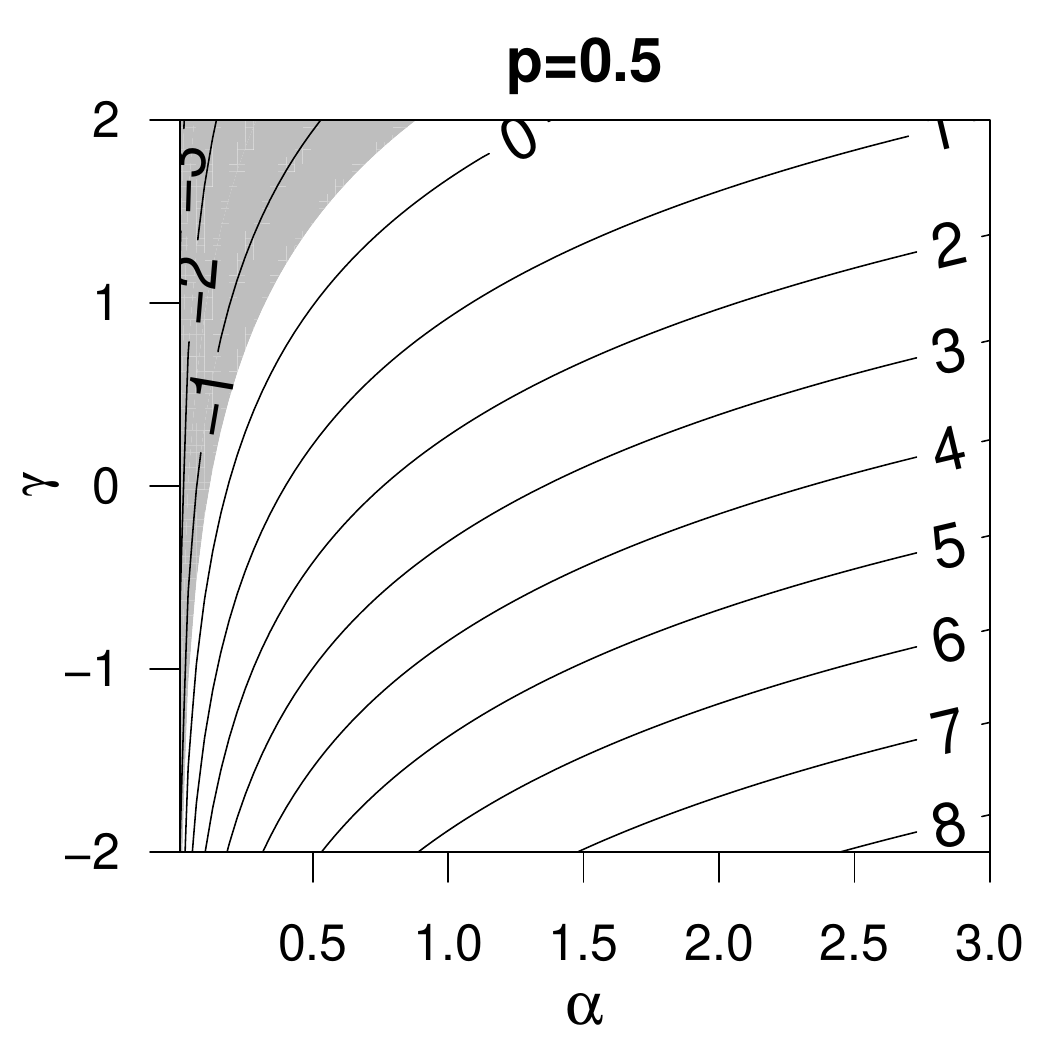}
    \caption{Contour plots of $h_2^2 / (\alpha h_1)^2$ on a log scale under $p = 0.05$, $0.3$, and $0.5$}
    \label{fig:Vcont}
\end{figure}
% The gray areas in the top-left corner represent regions where $h_2^2 / (\alpha h_1)^2 < 2/3$. This phenomenon can be interpreted using the properties of the gamma process: its mean is $e^{\gamma}t$, and its coefficient of variation (CV), defined as the ratio of the standard deviation to the mean, is $t/\sqrt{\alpha}$. The CV reflects the proportion of noise relative to information. As the CV increases more rapidly (i.e., when $\alpha$ decreases), the optimal $\tau$ tends to suggest more frequent inspections, since the uncertainty grows with time. Similarly, as the drift $\gamma$ increases, the optimal $\tau$ also recommends more frequent inspections. This is consistent with practical ADT settings, where higher stress levels are typically measured more frequently than lower stress levels so that we can get more information before the degradation path touches the threshold. On the other hand, for A-optimality, since it is not directly related to lifetime estimation, the drift parameter $\gamma$ does not affect the optimal inspection times. The inspection times tend to be allocated at larger values when the proportion of noise relative to information is higher. This can also be 發現 that when focus on prediction of lifetime or estimation of parameters will get the total different strategies.
The gray areas in the top-left corner represent regions where $h_2^2 / (\alpha h_1)^2 < 2/3$, and these areas slightly expand as $p$ increases.
We consider the coefficient of variation (CV) and mean of gamma process, which are $t/\sqrt{\alpha}$ and $e^\gamma t$. Since the CV represents the ratio of noise relative to information, the CV increases more rapidly with time when $\alpha$ is small, leading to greater uncertainty. In such cases, the optimal $\tau$ tends to suggest more frequent inspections. Similarly, as the drift parameter $e^\gamma$ increases, the optimal design also favors more frequent inspections. This consists with practical ADT settings, where higher stress levels are typically inspected more often to gather sufficient information before degradation reaches the threshold.
\begin{remark}
    When $m=1$, the results can be applied to destructive degradation test.
\end{remark}
% In contrast, under A-optimality, which is not directly related to lifetime estimation, the parameter $\gamma$ does not affect the optimal inspection times. When the proportion of noise relative to information is higher, the inspection times tend to be allocated at larger values.

\subsection{Optimal designs when $n$ and $T$ are fixed}
When $n$ and $T$ are fixed, the optimization problem for three criteria can all be simplified into the same problem.
\begin{equation*}
    \begin{aligned}
    \max\limits_{\tau} \quad &\tau\psi_1(\alpha\tau),\\
    \text{s.t.}\quad&\Delta t\leq\tau\leq T.
    \end{aligned}
\end{equation*}
Then, we have the following theorem.
\begin{theorem}\label{thm:type-I T fixed}
    The $\tau\psi_1(\alpha\tau)$ is a decreasing function. Hence, the optimal inspection time is $\Delta t$.
\end{theorem}
When $T$ is fixed, there is a trade-off between $m$ and $\tau$. Theorem~\ref{thm:type-I T fixed} suggests that the number of inspections is more influential than the inspection interval. In other words, when the termination time is fixed, more frequent inspections provide more information. In practice, conducting an inspection may incur a certain cost. The next scenario introduces a general cost constraint to better reflect real-world applications.

\subsection{Optimal designs with a cost constraint}
%As mentioned in Section 1, the limiting budget results in a cost constraint when conducting a degradation test. 
The cost constraint considered in this study is
\begin{equation}\label{eq:total_cost}
TC(n,m,\tau) = C_{it}n + C_{mea}nm + C_{op}m\tau \leq C_b,
\end{equation}
where $C_{it}$ is the unit cost of a test unit, $C_{mea}$ is the unit cost per inspection, $C_{op}$ is the operational cost per unit time, and $C_b$ is the total budget. Note that the optimal design always attains equality in this constraint, as the objective function favors a larger $n$. This cost constraint has been used in several literature \citep{tseng2009optimal,lee2020global, peng2021profile, cheng2024optimal}. Without loss of generality, we can simplify the problem by assuming $C_b=1$.

Finally, the optimization problem is 
\begin{equation*}
    \label{eq:criterionei}
    \begin{aligned}
    \min\limits_{\zeta=(n,m,\tau)} \quad &\phi(\zeta),\\
        \text{s.t.}\quad&TC(n,m,\tau)= 1,\\
        &\tau \geq \Delta t,\\
        &n,m\geq 1.
    \end{aligned}
\end{equation*}
Note that, the cost constraint will naturally restrict $\tau\leq(1-C_{it}-C_{mea})/C_{op}$. To simplify the notation,
% Similar to Theorem \ref{theorem: boundary search}, we can show that optimization occurs only when the cost constraint equality is hold.
% \begin{theorem}
%     $\phi(\zeta)$ attains its optimal value when $C_{it} n+C_{mea} m n+ C_{op} m\tau = 1$
% \label{theorem:EI boundary}
% \end{theorem}
% \noindent The proof is given in Appendix \ref{appendix:EI boundary}.
% Let $\tau_n(m)=(1-C_{it}n-C_{mea}nm)/(C_{op}m)$. Following the same procedure in section 3.3, we provide an algorithm to determine the integer optimal type-II design.
% \begin{algorithm}[H]
% \caption{Optimal type-II design}
% \label{alg:Boundary Search ei}
% \KwIn{$C_{it}, C_{mea}, C_{op}$}
% $n_{u} \gets \lfloor (1-C_{op} \Delta t) / (C_{it}+C_{mea}) \rfloor$\;
% \For{$n$ from 1 to $n_{u}$}{
%     $m_u \gets  \lfloor(1-C_{it} n)/ (C_{op}\Delta t +C_{mae} n)\rfloor$\;
%     $m_n=\arg\min_{m\in[1,m_u]}\phi(n,m,\tau_n(m))$ \;
%     $m_n^*=\arg\min_{m\in\{\lfloor m_n\rfloor,\lfloor m_n\rfloor+1\}}\phi(n,m,\tau_n(m))$\;
%     \eIf{$n=1$}{
%         $(n^*,m^*,\tau^*)\gets (n,m^*_{n},\tau_{n}(m_n^*))$ \;
%     }{
%         \If{$\phi(n,m_n^*,\tau_n(m_n^*))<\phi(n-1,m_{n-1}^*,\tau_{n-1}(m_{n-1}^*))$}{
%             $(n^*,m^*,\tau^*)\gets (n,m^*_{n},\tau_{n}(m_n^*))$ \;
%         }
%     }
% }
% \end{algorithm}
% Note that a general approach for approximate optimal type-II design is discussed in \cite{peng2021profile}. However, their discussion lacks detailed analysis of gamma degradation tests. To address this, we present the following theorem to provide further insights into this issue.
let
\begin{align*}
\varphi_D(\tau)&=\frac{1}{2}\frac{d }{d \tau}\log(\alpha \tau^3 \psi_1(\alpha \tau) - \tau^2),\\
\varphi_A(\tau)&=-\frac{d }{d \tau}\log\left[\frac{1}{\tau^2 \psi_1(\alpha \tau) - \frac{\tau}{\alpha}} + \frac{1}{\alpha \tau} \right],\\
\varphi_V(\tau)&=-\frac{d }{d \tau}\log\left[\frac{h_1^2}{\tau^2 \psi_1(\alpha \tau) - \frac{\tau}{\alpha}} + \frac{h_2^2}{\alpha \tau} \right],
\end{align*}
and
\begin{equation*}
K_1(\tau)=\left(\frac{C_{mea}}{C_{op}}+\tau\right)^{-1}, \quad
K_2(\tau) = \left(\frac{1}{C_{op}}-\tau\right)^{-1}, \quad
K_3(\tau) = \left(\tau\sqrt{1+\frac{C_{mea}}{C_{op}C_{it}\tau}}\right)^{-1}.
\end{equation*}
The following theorem provide the sufficient conditions for optimal type-I designs under the cost constraint.

\begin{theorem}
Given $C_{it}$, $C_{mea}$, $C_{op}$, $\Delta t$, $\eta$, $p$, $\alpha$ and $\gamma$, the optimal type-I design can be divided into eight cases:
\begin{enumerate}[leftmargin=*]
\item Let $\tau_1$ be the unique solution of 
\begin{equation}\label{eq:opEIcase2}
    \varphi(\tau)=K_1(\tau)
\end{equation}
If either $1 - 2C_{it} - C_{mea} > 0 \text{ and } \frac{C_{it}C_{mea}}{C_{op}(1 - 2C_{it})} > \tau_1>\Delta t
\text{ or }
1 - 2C_{it} - C_{mea} \leq 0 \text{ and } \frac{1 - C_{it} - C_{mea}}{C_{op}} > \tau_1>\Delta t,$
% $\begin{cases}
% 1 - 2C_{it} - C_{mea} \geq 0 \text{ and } \frac{C_{it}C_{mea}}{C_{op}(1 - 2C_{it})} > \tau_s>\Delta t, \\
% \text{or}\\
% 1 - 2C_{it} - C_{mea} < 0 \text{ and } \frac{1 - C_{it} - C_{mea}}{C_{op}} > \tau_s>\Delta t,
% \end{cases}$
then the optimal type-I design is $n^* = 1$, $m^* = \frac{1 - C_{it}}{C_{mea} + C_{op}\tau_1}$, and $\tau^*=\tau_1$.

\item Let $\tau_2$ be the unique solution of 
\begin{equation*}
    \varphi(\tau)=K_2(\tau).
\end{equation*}
If $1 - 2C_{it} - C_{mea} > 0$ and $\frac{1 - C_{it} - C_{mea}}{C_{op}} > \tau_2 > \max\left(\frac{C_{it}}{C_{op}(C_{mea} + 2C_{it})},\Delta t\right)$, then the optimal type-I design is $n^* = \frac{1 - C_{op}\tau_2}{C_{it} + C_{mea}}$, $m^* = 1$, and $\tau^*=\tau_2$.

\item Let
\begin{equation*}
     n(\tau)=\frac{-C_{it}+\sqrt{C_{it}^2+\frac{C_{mea}C_{it}}{C_{op}\tau}}}{\frac{C_{mea}C_{it}}{C_{op}\tau}},
     m(\tau)=\frac{-C_{it}+\sqrt{C_{it}^2+\frac{C_{mea}C_{it}}{C_{op}\tau}}}{C_{mea}},
\end{equation*}
and $\tau_3$ be the unique solution of
\begin{equation*}
    \varphi(\tau)=K_3(\tau).%\frac{1}{n(\tau)C_{mea}/C_{op}+\tau}=
\end{equation*}
If $1 - 2C_{it} -C_{mea}> 0$ and $ \frac{C_{it}}{C_{op}(C_{mea} + 2C_{it})} > \tau_3 >max\left(\frac{C_{it}C_{mea}}{C_{op}(1 - 2C_{it})},\Delta t\right)$, then the optimal type-I design is $n^*=n(\tau_3)$, $m^*=m(\tau_3)$ and $\tau^*=\tau_3$.

\item If $\varphi\left(\frac{1 - C_{mea} - C_{it}}{C_{op}}\right) > \max \left( \frac{C_{op}}{C_{it} + C_{mea}}, \frac{C_{op}}{1 - C_{it}} \right)$, then the optimal type-I design is $n^* = m^* = 1$, and $\tau^* = \frac{1 - C_{mea} - C_{it}}{C_{op}}$.

\item If $\varphi(\Delta t)<K_1(\Delta t)$
and either $1 - 2C_{it} - C_{mea} > 0 \text{ and } \frac{C_{it}C_{mea}}{C_{op}(1 - 2C_{it})} >\Delta t
\text{ or }
1 - 2C_{it} - C_{mea} \leq 0,$
then the optimal type-I design is $n^* = 1$, $m^* = \frac{1 - C_{it}}{C_{mea} + C_{op}\Delta t}$, and $\tau^*=\Delta t$.

\item If $\varphi(\Delta t)<K_2(\Delta t),$
$1 - 2C_{it} - C_{mea} > 0$ and $\Delta t>\frac{C_{it}}{C_{op}(C_{mea} + 2C_{it})}$, then the optimal type-I design is $n^* = \frac{1 - C_{op}\Delta t}{C_{it} + C_{mea}}$, $m^* = 1$, and $\tau^*=\Delta t$.

\item If $\varphi(\Delta t)<K_3(\Delta t),$
$1 - 2C_{it} -C_{mea}> 0$ and $ \frac{C_{it}}{C_{op}(C_{mea} + 2C_{it})} > \Delta t >\frac{C_{it}C_{mea}}{C_{op}(1 - 2C_{it})}$, then the optimal type-I design is $n^*=n(\Delta t)$, $m^*=m(\Delta t)$ and $\tau^*=\Delta t$.

\item If $C_{it}+C_{mea}+C_{op}\Delta t=1$, then the optimal type-I design is $n^* = m^* = 1$, and $\tau^* = \Delta t$.
\end{enumerate}
\label{thm:EI exact optimal design}
\end{theorem}

The proof is given in the supplementary material. Theorem \ref{thm:EI exact optimal design} provides a universal rule for all three criteria by simply replacing $\varphi$ with the corresponding $\varphi_D$, $\varphi_A$, and $\varphi_V$. Although Theorem \ref{thm:EI exact optimal design} looks complicated, it can be divided into three scenarios: $n=m=1$, $1-2C_{it}-C_{mea}\leq 0$, and $1-2C_{it}-C_{mea}> 0$.

\begin{itemize}[leftmargin=*]
\item For $n=m=1$, corresponding to cases 4 and 8, we check whether $\varphi\left(\frac{1-2C_{it}-C_{mea}}{C_{op}}\right)>\max\left(\frac{C_{op}}{C_{it}+C_{mea}},\frac{C_{op}}{1-C_{it}}\right)$ or $C_{it}+C_{mea}+C_{op}\Delta t=1$. Since both cases yield only one observation, they are not applicable in real applications. We recommend adjusting the budget to allow for more observations.     

\item If the optimal design does not satisfy the scenario $n=m=1$, we then check the case $1-2C_{it}-C_{mea}\leq0$. Clearly, under this condition, the optimal design always results in $n=1$, corresponding to cases 1 and 5. If a solution exists for equation (\ref{eq:opEIcase2}), it belongs to case 1; otherwise, it falls under case 5.

\item Finally, if neither $n=m=1$ nor $1-2C_{it}-C_{mea}\leq 0$ holds, we move on to the scenario $1-2C_{it}-C_{mea} > 0$. In this scenario, the optimal design is determined by solving the equations
\begin{equation*}
    \varphi(\tau)=K(\tau),
\end{equation*}
where
\begin{equation*}
K(\tau)=\left\lbrace\begin{matrix}
    K_1(\tau),& 0<\tau\leq\frac{C_{it}C_{mea}}{C_{op}(1-2C_{it})},\\
    K_3(\tau),& \frac{C_{it}C_{mea}}{C_{op}(1-2C_{it})}<\tau \leq \frac{C_{it}}{C_{op}(C_{mea}+2C_{it})},\\
    K_2(\tau),& \frac{C_{it}}{C_{op}(C_{mea}+2C_{it})}<\tau \leq \frac{1-C_{it}-C_{mea}}{C_{op}}.
\end{matrix}\right.
\end{equation*}
It is easy to show that $K(\tau)$ is a continuous function decreasing from $\left[0,\frac{C_{it}}{C_{op}(C_{mea}+2C_{it})}\right]$ and increasing from $\left(\frac{C_{it}}{C_{op}(C_{mea}+2C_{it})},\frac{1-C_{it}-C_{mea}}{C_{op}}\right]$ .
If the solutions 
$\tau^*\in\left[\Delta t,\frac{1-C_{it}-C_{mea}}{C_{op}}\right]$, then the optimal design is determined by comparing with the three cost indices:
\begin{equation*}
\frac{C_{it}C_{mea}}{C_{op}(1-2C_{it})} < \frac{C_{it}}{C_{op}(C_{mea}+2C_{it})} < \frac{1-C_{it}-C_{mea}}{C_{op}},
\end{equation*}
where the order is provided by $1-2C_{it}-C_{mea}>0$. More specifically, 
\begin{align*}
\tau^*&<\frac{C_{it}C_{mea}}{C_{op}(1-2C_{it})}\Rightarrow \text{ case 1},\\
\frac{C_{it}C_{mea}}{C_{op}(1-2C_{it})}<\tau^*& < \frac{C_{it}}{C_{op}(C_{mea}+2C_{it})}\Rightarrow \text{ case 3},\\
\frac{C_{it}}{C_{op}(C_{mea}+2C_{it})}<\tau^*& < \frac{1-C_{it}-C_{mea}}{C_{op}}\Rightarrow \text{ case 2}.
\end{align*}
If $\tau^*\notin\left[\Delta t,\frac{1-C_{it}-C_{mea}}{C_{op}}\right]$ then 
\begin{align*}
\Delta t&<\frac{C_{it}C_{mea}}{C_{op}(1-2C_{it})}\Rightarrow \text{ case 5},\\
\frac{C_{it}C_{mea}}{C_{op}(1-2C_{it})}<\Delta t& < \frac{C_{it}}{C_{op}(C_{mea}+2C_{it})}\Rightarrow \text{ case 7},\\
\frac{C_{it}}{C_{op}(C_{mea}+2C_{it})}<\Delta t& < \frac{1-C_{it}-C_{mea}}{C_{op}}\Rightarrow \text{ case 6}.
\end{align*}
Fig. \ref{fig:optcasedetermine} illustrates how $\varphi(\tau)$, $K(\tau)$, and $\Delta t$ determine the case of optimal design. The intersections of $\varphi(\tau)$ with $K(\tau)$ and the position of $\tau^*$ relative to $\Delta t$ categorize the cases.
\begin{remark}
    According to the properties of $K(\tau)$, it is difficult to analytically prove that $\varphi(\tau) = K(\tau)$ has a unique root. However, in several applications, the existence of a unique root is readily observed, as demonstrated in the illustration section.
\end{remark}
\begin{figure}[ht!]
    \centering
    \includegraphics[width=0.8\textwidth]{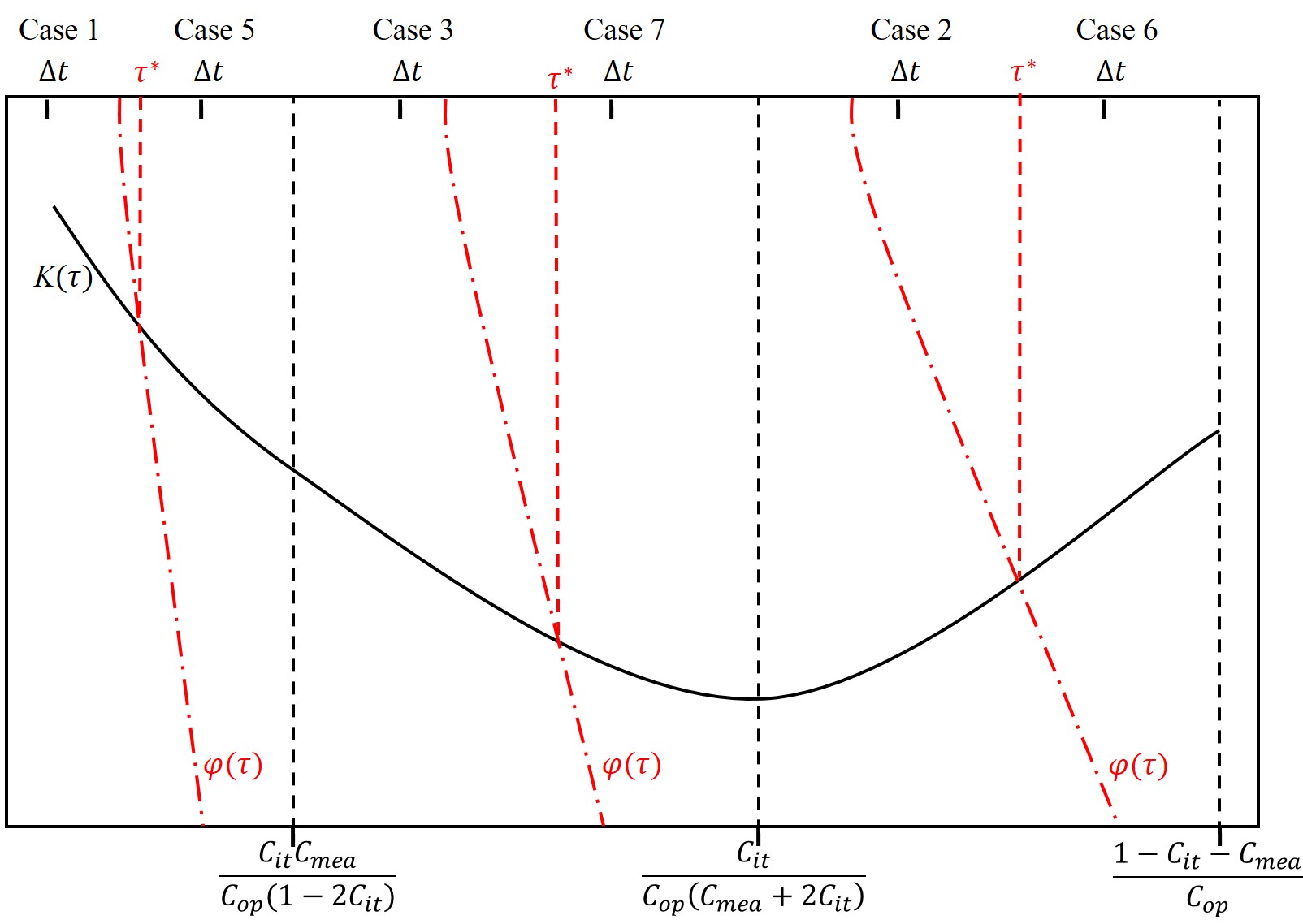}
    \caption{ Illustration of how $\varphi(\tau)$, $K(\tau)$, and $\Delta t$ determine the optimal design case.}
    \label{fig:optcasedetermine}
\end{figure}
% \begin{figure}
%     \centering
%     \includegraphics[width=0.8\textwidth]{opt_Icase576.pdf}
%     \caption{Caption}
%     \label{fig:enter-label}
% \end{figure}
\end{itemize}

% To understand how the parameter and cost affect the optimal design. Figures \ref{fig:varphi_D_A} and \ref{fig:varphi_V} plot the curves of $\varphi(\tau)$ under different parameter settings. Note that $\varphi_D$ and $\varphi_A$ are only related to $\alpha$. Figure \ref{fig:K(tau)} plot the curves of $K_1(\tau),K_2(\tau),K_3(\tau)$ which only relate to cost.

% Theorem \ref{thm:EI exact optimal design} shows that the cost index
% \begin{equation*}
%     \frac{1-C_{it}-C_{mea}}{C_{op}}>\frac{C_{it}}{C_{op}(C_{mea}+2C_{it})}>\frac{C_{it}C_{mea}}{C_{op}(1-2C_{it})},
% \end{equation*}
% is crucial in determining the optimal design. Given a specific cost, the approximate optimal design can be identified based on which cost index interval the value of $\tau^*$ falls into.

\section{Optimal type-II designs}
In this section, we additionally determine a decision variable: the inspection interval $\t=(\Delta t_1,\ldots,\Delta t_m)$, which is rarely discussed in the literature. Under this scenario, the optimal design problem is
\begin{equation*}
    \label{eq:criterionoi}
    \begin{aligned}
    \min\limits_{\zeta= (n, m, T,\t)} \quad &\phi(\zeta),\\
        \text{s.t.}\quad&TC(n,m,T)= 1,\\
        &\Delta t_j \geq \Delta t,\\
        &n,m\geq 1.
    \end{aligned}
\end{equation*}
where $\phi_D$, $\phi_A$ and $\phi_V$ are redefined as
\begin{equation*}
\phi_D(\zeta) = \frac{1}{n^2T\left\lbrace\alpha\sum\limits_{j=1}^m\left[\Delta t_j^2\psi_1(\alpha\Delta t_j)\right]-T\right\rbrace},
\end{equation*}
\begin{equation*}
\phi_A(\zeta) = \frac{1}{n} \left\lbrace\frac{1}{\sum\limits_{j=1}^m\left[\Delta t_j^2\psi_1(\alpha\Delta t_j)\right]-\frac{T}{\alpha}}+\frac{1}{\alpha T}\right\rbrace.
\end{equation*}
\begin{equation*}
\phi_V(\zeta)
= \frac{1}{n} \left\lbrace\frac{h_1^2}{\sum\limits_{j=1}^m\left[\Delta t_j^2\psi_1(\alpha\Delta t_j)\right]-\frac{T}{\alpha}}+\frac{h_2^2}{\alpha T} \right\rbrace.
\end{equation*}

\subsection{Optimal aperiodic inspection time}
When $n$, $m$ and $T$ are given, the optimization problem of three criteria can be simplified as:
\begin{equation}
\begin{aligned}
    &\max_{\t}~g_1(\t)=\sum_{j=1}^{m} \bigg[ (\alpha \Delta t_{j})^2 \psi_1(\alpha \Delta t_{j}) \bigg]\\
    &\text{s.t.}~\sum_{j=1}^{m} \Delta t_{j}=T,\Delta t_j\geq \Delta t.
\end{aligned}
\label{eq:optimization_problem_part1}
\end{equation}
Problem (\ref{eq:optimization_problem_part1}) is a standard majorization optimization problem. By applying majorization theory \citep{marshall2011inequalities}, we have the following theorem.
\begin{theorem}
    In optimization problem (\ref{eq:optimization_problem_part1}), the maximum occurs at $(T-(m-1)\Delta t,\Delta t, \dots, \Delta t)$ for any ordering and the minimum occurs at $(T/m,\ldots,T/m)$.
    \label{theorem:max_at_tj}
\end{theorem}
\noindent The proof is straightforward. According to Theorem 4.14 in \cite{alzer2004sharp}, $x^2\psi_1(x)$ is strictly convex, which implies that $g_1$ is Schur-convex. Thus, the result follows.

Theorem \ref{theorem:max_at_tj} shows that the periodic inspection time is actually the least efficient. This finding has important implications for gamma degradation tests, where inspections are often taken at equal inspection intervals for simplicity. Our result suggests that such an approach may lead to suboptimal data collection, potentially reducing the accuracy of degradation modeling.
\begin{remark}
    The ordering of $(T-(m-1)\Delta t, \Delta t, \ldots, \Delta t)$ does not affect the performance of the optimal design. This result follows from the stationary assumption in the gamma process, which ensures that the amount of information obtained at each time point remains constant. Consequently, the order of inspection intervals has no impact on the efficiency of the optimal design.
\end{remark}

Building on this insight, we now examine the optimal design to further determine $n$, $m$ and $T$.
\subsection{Optimal designs with $C_{mea}=0$}
We first consider a simple case where $C_{mea} = 0$, a scenario also commonly encountered in real applications.
% Although the original cost constraint includes the inspection cost, in some situations, this cost are absent.
For example, in a rechargeable lithium-ion battery degradation test \citep{wang2019end}, the quality characteristics (battery capacity) is automatically recorded by the machine, so $C_{mea}$ is absent.

In this scenario, the budget constraints are defined as:
\begin{equation*}
C_{it} n+ C_{op} T = 1,
\end{equation*} 
which does not depend on $m$. Hence, we can first determine $m$ when $n$ and $T$ are given. Under this scenario, the optimization problem for three criteria can be simplified as
\begin{equation}    
\begin{aligned}
    \max_{m} \quad &g_2(m)= (m-1)(\alpha\Delta t)^2 \psi_1(\alpha \Delta t)+[\alpha(T-(m-1) \Delta t)]^2 \psi_1(\alpha(T-(m-1) \Delta t)) \\
    \textrm{s.t.} \quad & 1 \leq m \leq T/\Delta t
\end{aligned}
\label{eq:optimization_problem_part2}
\end{equation} 

\noindent Then, we have the following theorem.
\begin{theorem}
    In optimization problem (\ref{eq:optimization_problem_part2}), the maximum occurs at $m=T/\Delta t$.
    \label{theorem:max_at_m}
\end{theorem}
\begin{proof}
    Let $x=\alpha \Delta t$ and $y=\alpha(T-(m-1)\Delta t)$,
    \begin{equation*}
        \frac{d g_2(m)}{d m}= x \bigg[x \psi_1(x)-2 y \psi_1(y)-y^2 \psi_2(y)\bigg]>0,
    \end{equation*} 
    provided by Lemma \ref{lemma:alzer20012004}.
    Hence, $g_2$ is an increasing function, and its maximum occurs at $m=T/\Delta t$.
\end{proof}
Theorem \ref{theorem:max_at_m} demonstrates that with no inspection cost, the optimal strategy is to maximize number of inspections $(m=T/\Delta t)$, leading to evenly spaced intervals of $\Delta t$ over the total testing time $T$. This reduces a type-II design to a type-I design, highlighting that frequent inspections maximize information collection when inspection cost is absent.

Next, given $m=T/\Delta t$, we then further determine $n$ and $T$. The optimization problem for the three criteria can be simplified as follows:
\begin{equation} 
\begin{aligned}
     \max \quad  &nT \\   
     \textrm{s.t.} \quad & C_{it} n+ C_{op} T = 1.
     \label{eq:optimization_problem_part3}
\end{aligned}
\end{equation}
Direct computation yields $n^*=\frac{1}{2C_{it}}$ and $T^*=\frac{1}{2C_{op}}$. Hence, when there is no inspection cost, the optimal type-II design will be $(n^*,m^*,T^*)=\left(\frac{1}{2C_{it}},\frac{1}{2C_{op}\Delta t},\frac{1}{2C_{op}}\right)$.

\subsection{Optimal designs with $C_{mea}\neq0$}
Under the optimal inspection time, the optimization problem is
\begin{equation}\label{eq:criteriononlytj}
    \begin{aligned}
     \min\limits_{\zeta=(n,m,T)} \quad &\phi(\zeta),\\
        \text{s.t.} \quad &C_{it}n+C_{mea}nm+C_{op}T = 1,\\
        &T\geq m\Delta t,\\
        &n,m\geq 1,
    \end{aligned}
\end{equation}
% \begin{equation*}
%     \phi_D(\zeta)
%     =\frac{1}{n^2}\left\lbrace\frac{1}{\alpha T \lbrace(m-1)\Delta t^2\psi_1(\alpha\Delta t)+(T-(m-1)\Delta t)^2\psi_1[\alpha(T-(m-1)\Delta t)]\rbrace-T^2}\right\rbrace,
% \end{equation*} 
% \begin{equation*} 
%     \phi_A(\zeta)
%     =\frac{1}{n} \bigg\lbrace\frac{1}{(m-1)\Delta t^2\psi_1(\alpha\Delta t)+(T-(m-1)\Delta t)^2\psi_1[\alpha(T-(m-1)\Delta t)]-T/\alpha}+\frac{1}{\alpha T}\bigg\rbrace,
% \end{equation*}
% \begin{equation*}
%      \phi_V(\zeta)
%     =\frac{1}{n} \bigg\lbrace\frac{h_1^2}{(m-1)\Delta t^2\psi_1(\alpha\Delta t)+(T-(m-1)\Delta t)^2\psi_1[\alpha(T-(m-1)\Delta t)]-T/\alpha}+\frac{h_2^2}{\alpha T}\bigg\rbrace,
% \end{equation*}
where
\begin{align*}
    &\phi_D(\zeta)=(n^2\varrho_D(m,T))^{-1},\\
    &\phi_{A,V}(\zeta)=(n\varrho_{A,V}(m,T))^{-1},\\
    &\varrho_D(m,T)=\alpha T \lbrace(m-1)\Delta t^2\psi_1(\alpha\Delta t)+(T-(m-1)\Delta t)^2\psi_1[\alpha(T-(m-1)\Delta t)]\rbrace-T^2,\\
    &\varrho_A(m,T)=\bigg\lbrace\frac{1}{(m-1)\Delta t^2\psi_1(\alpha\Delta t)+(T-(m-1)\Delta t)^2\psi_1[\alpha(T-(m-1)\Delta t)]-T/\alpha}+\frac{1}{\alpha T}\bigg\rbrace^{-1},\\
    &\varrho_V(m,T)=\bigg\lbrace\frac{h_1^2}{(m-1)\Delta t^2\psi_1(\alpha\Delta t)+(T-(m-1)\Delta t)^2\psi_1[\alpha(T-(m-1)\Delta t)]-T/\alpha}+\frac{h_2^2}{\alpha T}\bigg\rbrace^{-1}.
\end{align*}
Let
\begin{equation*}
    \varphi_m(m,T)=\left\lbrace\begin{matrix}
        \frac{1}{2}\frac{\partial}{\partial m}\log\varrho_D(m,T)\\
        \frac{\partial}{\partial m}\log\varrho_{A,V}(m,T)
    \end{matrix}\right.,
    \varphi_T(m,T)=\left\lbrace\begin{matrix}
        \frac{1}{2}\frac{\partial}{\partial T}\log\varrho_D(m,T)\\
        \frac{\partial}{\partial T}\log\varrho_{A,V}(m,T)
    \end{matrix}\right..
\end{equation*}
Then, The following theorem provide the sufficient conditions for optimal type-II designs.
\begin{theorem}
Given $C_{it}$, $C_{mea}$, $C_{op}$, $\Delta t$, $\eta$, $p$, $\alpha$ and $\gamma$, the optimal type-II design can be divided into eight cases:
\begin{enumerate}[leftmargin=*]
\item Let
\begin{equation*}
    T(m)=\frac{1-C_{mea}m-C_{it}}{C_{op}},
\end{equation*}
and $m_1$ be the solution of
\begin{equation*}
    \frac{\varphi_m(m,T(m))}{\varphi_T(m,T(m))}=\frac{C_{mea}}{C_{op}}.
\end{equation*}
If
\begin{equation*}
    \frac{C_{op}}{C_{it}+C_{mea}m_1}<\varphi_T(m_1,T(m_1)),
\end{equation*}
then the optimal type-II design is $n^*=1$, $m^*=m_1$ and $T^*=T(m_1)$.
\item Let
\begin{equation*}
    T(n)=\frac{1-n(C_{mea}+C_{it})}{C_{op}},
\end{equation*}
and $n_2$ be the solution of
\begin{equation*}
    n\varphi_T(1,T(n))=\frac{C_{op}}{C_{it}+C_{mea}}.
\end{equation*}
If
\begin{equation*}
    \varphi_m(1,T(n_2))<\frac{C_{mea}}{C_{it}+C_{mea}},
\end{equation*}
then the optimal type-II design is $n^*=n_2$, $m^*=1$ and $T^*=T(n_2)$.

\item Let
\begin{equation*}
    T(n,m)=\frac{1-C_{it}n-C_{mea}nm}{C_{op}},
\end{equation*}
and $n_3$ and $m_3$ be the solution of
\begin{equation*}
    n\varphi_T(m,T(n,m))=\frac{C_{op}}{C_{it}+C_{mea}m}, \varphi_m(m,T(n,m))=\frac{C_{mea}}{C_{it}+C_{mea}m}.
\end{equation*}
If $n_3>1$, $m_3>1$ and $T(n_3,m_3)>m_3\Delta t$, then the optimal type-II design is $n^*=n_3$, $m^*=m_3$ and $T^*=T(n_3,m_3)$.

\item Let
\begin{equation*}
    T_4 = \frac{1 - C_{mea} - C_{it}}{C_{op}}.
\end{equation*}
If
\begin{equation*}
    \frac{C_{op}}{C_{it}+C_{mea}}<\varphi_T(1,T_4),
\end{equation*}
\begin{equation*}
    \frac{\varphi_m(1,T_4)}{\varphi_T(1,T_4)}<\frac{C_{mea}}{C_{op}},
\end{equation*}
then the optimal type-II design is $n^* = m^* = 1$, and $T^* = \frac{1 - C_{mea} - C_{it}}{C_{op}}$.

\item Let
\begin{equation*}
    m_5=\frac{1-C_{it}}{C_{mea}+C_{op}\Delta t}.
\end{equation*}
If
\begin{equation*}
    \frac{C_{mea}+\Delta tC_{op}}{C_{it}+C_{mea}m_5}<\varphi_m(m_5,m_5\Delta t)+\Delta t\varphi_T(m_5,m_5\Delta t),
\end{equation*}
then
then the optimal type-II design is $n^* = 1$, $m^* =\frac{1-C_{it}}{C_{mea}+C_{op}\Delta t}$, and $T^* = m^*\Delta t$.
\item Let
\begin{equation*}
    n_6=\frac{1-C_{op}\Delta t}{C_{it}+C_{mea}}.
\end{equation*}
If
\begin{equation*}
    n_6[\varphi_m(1,\Delta t)+\Delta t\varphi_T(1,\Delta t)]<\frac{C_{op}\Delta t+C_{mea}n_6}{C_{it}+C_{mea}},
\end{equation*}
\begin{equation*}
    n_6\varphi_T(1,\Delta t)<\frac{C_{op}}{C_{it}+C_{mea}},
\end{equation*}
then the optimal type-II design is $n^* = \frac{1-C_{op}\Delta t}{C_{it}+C_{mea}}$, $m^* =1$, and $T^* = \Delta t$.

\item Let
\begin{equation*}
    n(m)=\frac{1-C_{op}m\Delta t}{C_{it}+C_{mea}m},
\end{equation*}
and $m_7$ be the solution of
\begin{equation*}
    n(m)[\varphi_m(m,m\Delta t)+\Delta t\varphi_T(m,m\Delta t)]=\frac{C_{op}\Delta t+C_{mea}n(m)}{C_{it}+C_{mea}m}.
\end{equation*}
If
\begin{equation*}
    n(m_7)\varphi_T(m_7,m_7\Delta t)<\frac{C_{op}}{C_{it}+C_{mea}m_7},
\end{equation*}
then the optimal type-II design is $n^* = n(m_7)$, $m^* =m_7$, and $T^* = m_7\Delta t$.
\item If $C_{it}+C_{mea}+C_{op}\Delta t=1$, then the optimal type-II design is $n^* = m^* = 1$, and $T^* = \Delta t$.
\end{enumerate}
\label{thm:OI exact optimal design}
\end{theorem}
The proof is similar to that of Theorem \ref{thm:EI exact optimal design} and therefore is omitted. Theorem \ref{thm:OI exact optimal design} offers a structured framework for determining the optimal type-II design under cost constraints. By systematically classifying the optimal design into eight cases, it provides actionable insights for real-world applications. Compared to existing literature, this theorem generalizes optimal gamma degradation test by simultaneously considering inspection interval. 
\begin{remark}
When the constraint $T = m \Delta t$ holds, as in cases 5-8, the optimal type-II design reduces to the optimal type-I design with $\tau = \Delta t$. Additionally, when the constraint $m = 1$ holds, as in cases 2, 4, 6, and 8, the design also reduces to a type-I design, and can be applied to destructive degradation tests.
\end{remark}

%However, comparing with theorem \ref{thm:EI exact optimal design}, the equations become more complicate and hence not easy to summarize some clear formulation.

\section{Illustrative examples}

\begin{example}
    We revisit the LED data from \cite{liao2004reliability} to illustrate how to determine the optimal type-I designs. As reported in \cite{lim2015optimum}, the parameter estimates are $\alpha = 0.065$ and $\gamma = -0.77$. The costs are $C_{it} = 3 \times 10^{-2}/\text{unit}$, $C_{mea} = 1.9 \times 10^{-3}/\text{inspection}$, and $C_{op} = 2.7 \times 10^{-3}/\text{hr}$, with $\Delta t = 5\text{hrs}$. The threshold is $\eta = 0.5$, and the interest is in the $p$th quantile, with $p = 0.1$. Under these settings, according to Theorem \ref{thm:EI exact optimal design}, the functions $K(\tau)$, $\varphi_D(\tau)$, $\varphi_A(\tau)$, and $\varphi_V(\tau)$, represented by the lines in solid, dotted, dashed, and dotdash, respectively, are plotted in Fig. \ref{fig:k_tau_ex1} with the x-axis and y-axis in log scale. The three cost indices are:
    \begin{equation*}
    \frac{C_{it}C_{mea}}{C_{op}(1 - 2C_{it})} = 0.0225, \quad \frac{C_{it}}{C_{op}(C_{mea} + 2C_{it})} = 179.5, \quad \frac{1 - C_{it} - C_{mea}}{C_{op}} = 358.6.
    \end{equation*}
    \begin{figure}
        \centering
        \includegraphics[width=\textwidth]{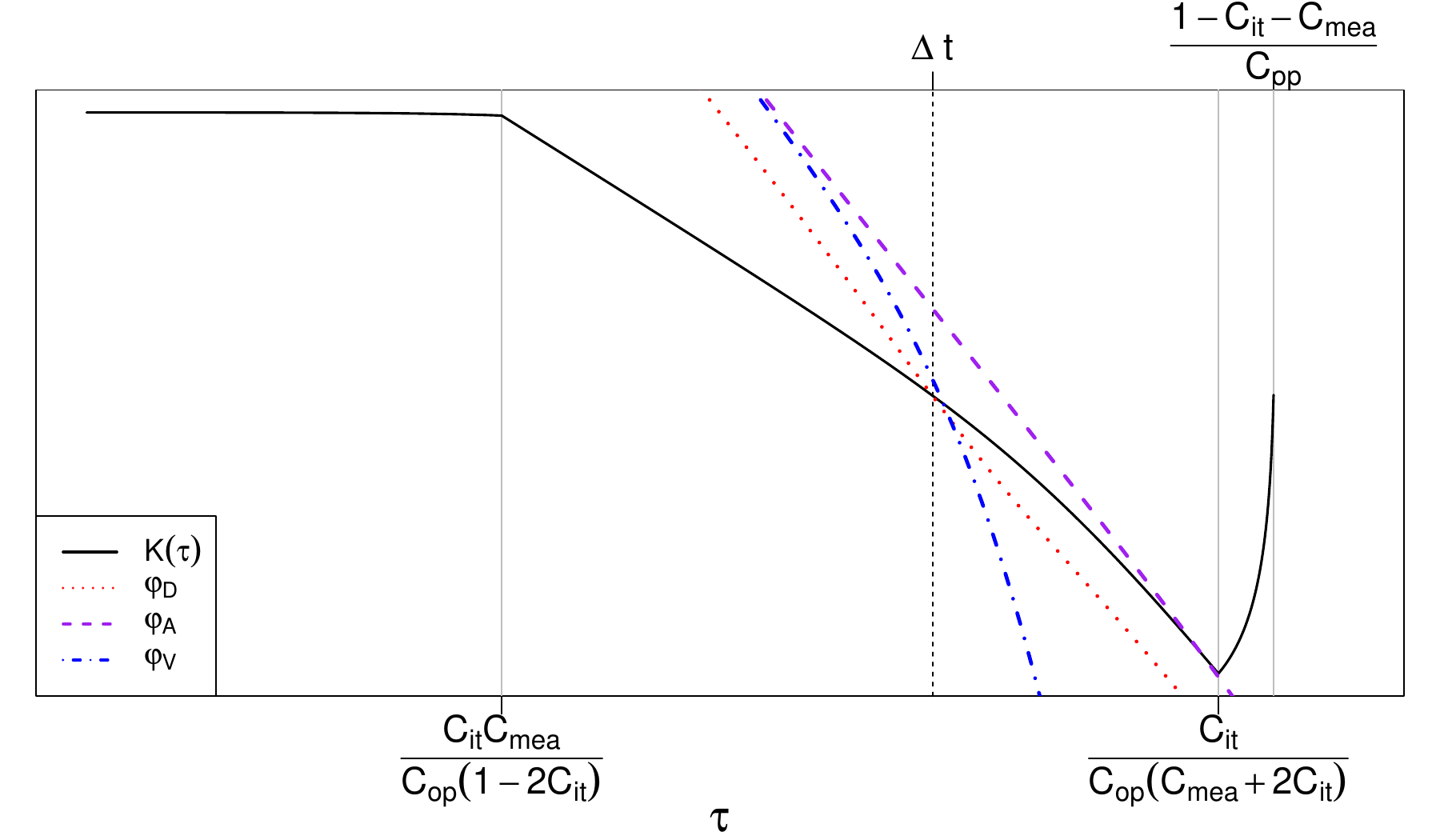}
        \caption{Illustration of $K(\tau)$ and $\varphi(\tau)$ under the setting of Example 1}
        \label{fig:k_tau_ex1}
    \end{figure}
    The roots of $\varphi(\tau) = K(\tau)$ are 1.34, 143.2, and 5.72 for D-, A-, and V-optimality, respectively, all of which lie in $\left(\frac{C_{it}C_{mea}}{C_{op}(1-2C_{it})},\frac{C_{it}}{C_{op}(C_{mea}+2C_{it})}\right)$. For D-optimality, the root is smaller than $\Delta t$, resulting in case 7. In contrast, the roots for A- and V-optimality are greater than $\Delta t$, resulting in case 3.
    Table \ref{table:lim2015} summarizes the optimal type-I designs. Additionally, we compare these results with the optimal type-II designs. To facilitate comparison, we define the relative efficiency of a design $\zeta$ as
    \begin{equation*}
        RE(\zeta^*_{II},\zeta)=\frac{\phi(\zeta^*_{II})}{\phi(\zeta)},
    \end{equation*}
    where $\zeta^*_{II}$ is the optimal type-II design.
    In Table \ref{table:lim2015}, the relative efficiencies of the optimal type-I designs are all close to 1. Thus, in this example, if a practitioner prefers a simpler experimental setup, the optimal type-I design can be used without significant loss of efficiency.
    \begin{table}[H]
    \caption{Comparison of optimal type-I and type-II designs for Example 1 under three criteria}
    \centering
        \begin{tabular}[c]{|c|c|c|c|c|c|c|c|}
        \hline
         Criterion & Time planning &$\phi(\zeta^*)$&RE &$n^\ast$&$m^\ast$&$T^\ast$&Case\\
        \hline
        \multirow{2}{*}{$D$}
          & Type-I & \multirow{1}{*}{$3.53\times10^{-7}$}&0.99 &\multirow{1}{*}{9.85}&\multirow{1}{*}{21.9}&\multirow{1}{*}{109.4}&Case 7\\
        &  Type-II &$3.48\times10^{-7}$&* & 10.9&16.6&122.5 &Case 3\\
        \hline
        \multirow{2}{*}{$A$}
          & Type-I & \multirow{1}{*}{$5.79\times10^{-3}$} &0.99
          &\multirow{1}{*}{16.0}&\multirow{1}{*}{1.24}&\multirow{1}{*}{178.2}& Case 3 \\
        & Type-II &$5.75\times10^{-3}$&*&15.8&1.35&179.5 &Case 3\\
        \hline    
        \multirow{2}{*}{$V$}
          & Type-I & \multirow{1}{*}{$2.47\times10^{-3}$} &0.98
          &\multirow{1}{*}{10.2}&\multirow{1}{*}{19.9}&\multirow{1}{*}{113.7}& Case 3\\
        & Type-II & $2.43\times10^{-3}$&* &10.6 &17.7&119.7&Case 3\\
        \hline
        \end{tabular}
    \label{table:lim2015}
    \end{table}
\end{example}

\begin{example}
   Although the optimal Type-I design does not result in significant loss of efficiency in Example 1, this example illustrates a situation where convenience leads to a sacrifice of efficiency. Consider the LED light degradation data from Table 6.3 in \cite{chaluvadi2008accelerated}. This dataset includes 12 samples tested under a current of 40mA. The quality characteristic (QC) for this data is light intensity, which gradually decreases over time. Since the dataset does not provide the light intensity at time $t = 0$, we assume the light intensity at $t = 0$ to be 90. The light intensity of each unit was inspected every 50 hours, with the test running for a total of 250 hours. A sample was considered to have failed when the light intensity dropped below $\eta = 50$. The left panel of Fig. \ref{fig:ex2data} shows the degradation path for the LED dataset.
    \begin{figure}
        \centering
        \includegraphics[width=0.4\textwidth]{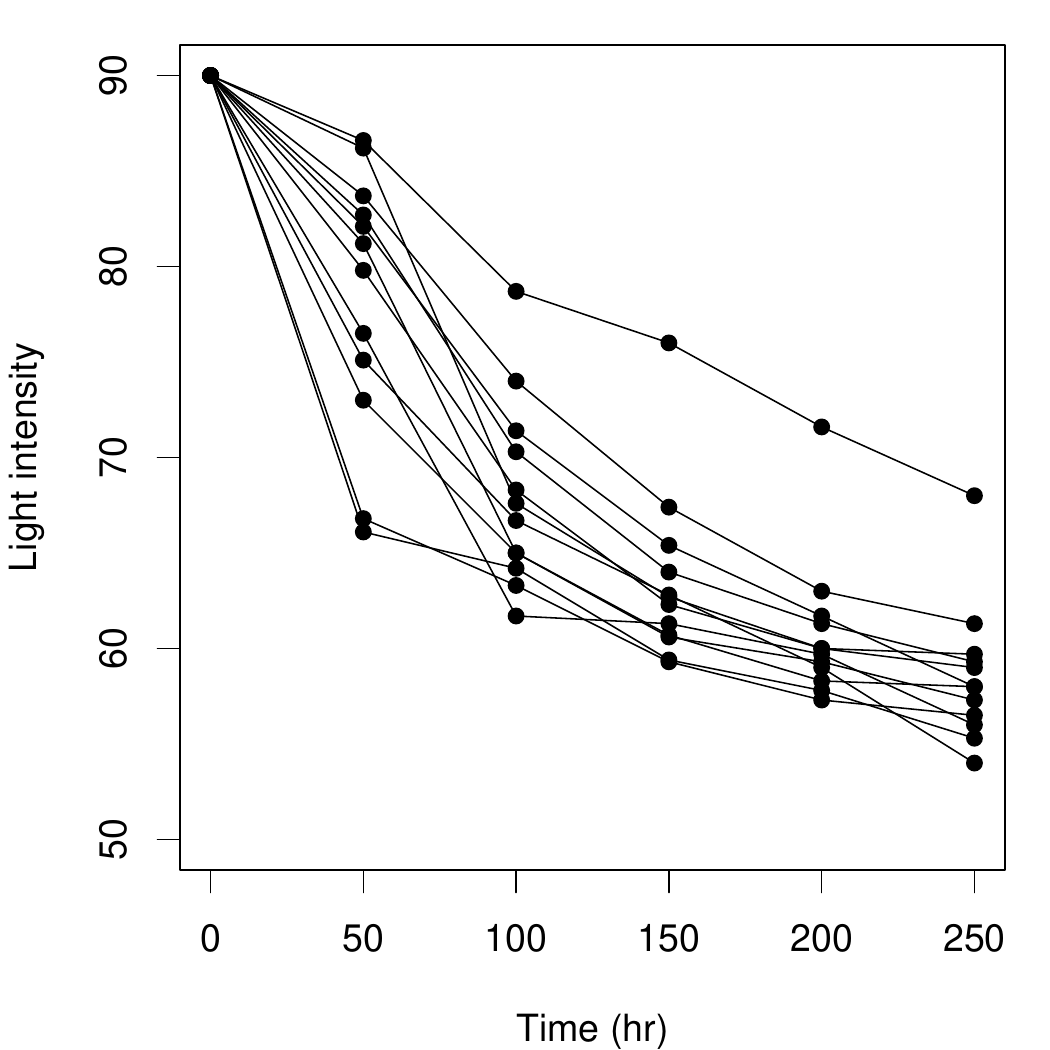}
        \includegraphics[width=0.4\textwidth]{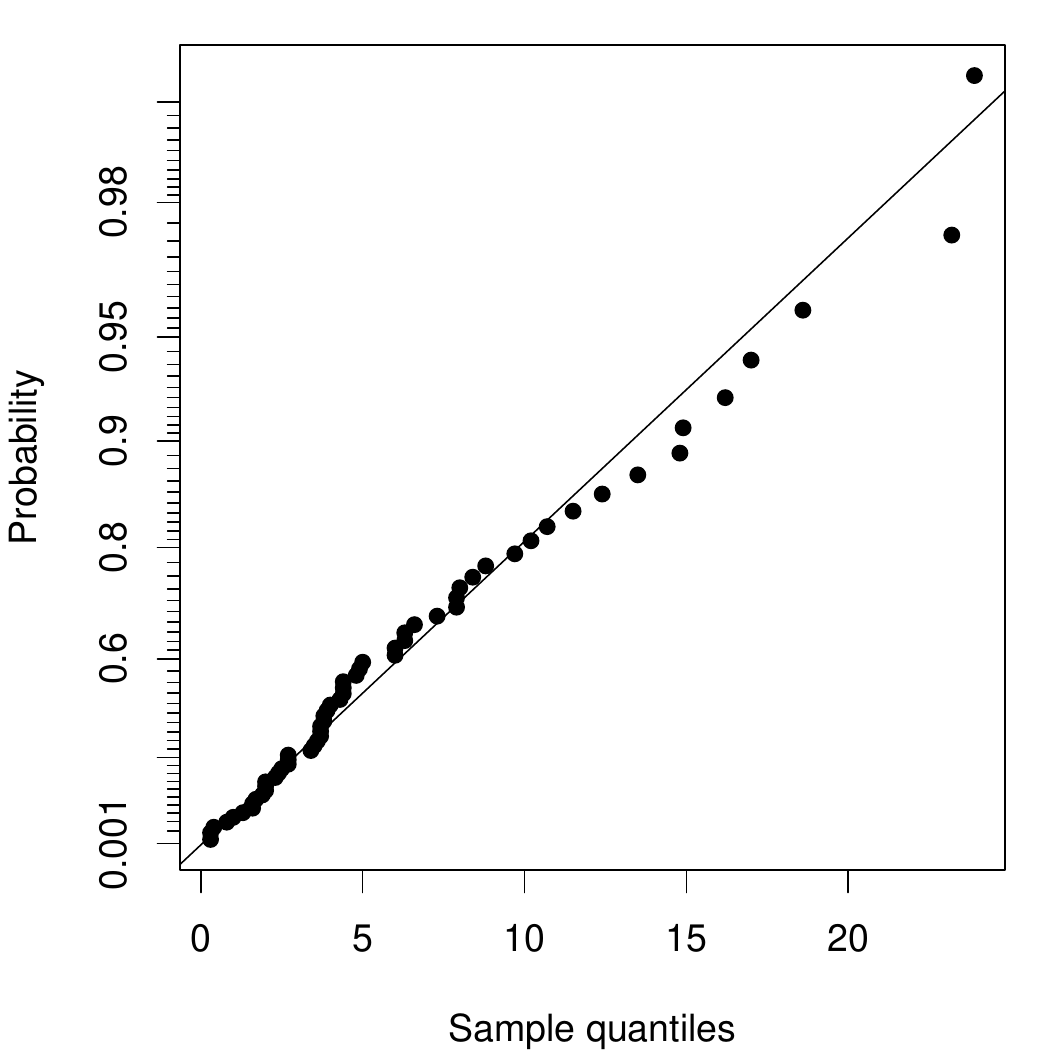}
        \caption{LED degradation data from \cite{chaluvadi2008accelerated} and gamma probability plot}
        \label{fig:ex2data}
    \end{figure}
    % with covariance matrix calculated from equation (\ref{eq:fisherinformationmatrix})
    % \begin{equation}\label{eq:LEDFIM}
    %     \begin{pmatrix}
    %         1.26\times10^{-5}&0\\
    %         0&1.52\times10^{-2}
    %     \end{pmatrix}.
    % \end{equation} 
    
    By maximizing equation (\ref{eq:loglikelihoodfunction}), the MLE of the parameters are $\hat{\alpha} = 0.028$ and $\hat{\gamma} = -2.073$. The corresponding variances, computed from the inverse of the Fisher information matrix, are $\text{Var}(\hat{\alpha}) = 2.18 \times 10^{-5}$ and $\text{Var}(\hat{\gamma}) = 1.18 \times 10^{-2}$.
    The probability plot is provided in the right panel Fig. \ref{fig:ex2data}, where the linearity suggests that the data is suitable for fitting with a gamma process. Additionally, the Anderson-Darling test was used to assess the goodness of fit, yielding a $p$-value of 0.4396. At a significance level of 0.05, we obtain the same conclusion as with the gamma probability plot.
    
    Next, we set $C_{it} = 7.56 \times 10^{-2}$/unit, $C_{mea} = 1.06 \times 10^{-3}$/inspection, $C_{op} = 1.17 \times 10^{-4}$/hr, and $\Delta t = 5$ hrs. In this setting, $C_b = 1$ corresponds to the original experimental configuration. For $V$-optimality, we assume $p = 0.05$. Table \ref{table:ch2008} presents the results of the original design, the optimal type-I, and type-II designs for the three criteria. Additionally, to compare with real applications, a grid search algorithm is used to find the optimal integer design for comparison.
    
    \begin{table}[H]
    \caption{Comparison of different optimal designs for Example 2 under three criteria}
    \centering
        \begin{tabular}[c]{|c|c|c|c|c|c|c|}
        \hline
         Criterion & Time planning &$\phi(\zeta^*)$&RE &$n^\ast$&$m^\ast$&$T^\ast$\\
        \hline
        \multirow{5}{*}{$D$}
          & Type-I & $1.082\times10^{-8}$& 0.75&3.82&104.55&2466\\
          & Type-I integer & $1.084\times10^{-8}$&0.75 &4&98&2411\\
        &  Type-II &$8.122\times10^{-9}$&* & 4.38&72.19 &2849 \\
        &  Type-II integer &$8.17\times10^{-9}$&0.99 & 4&82 &2991 \\
        &Original& $2.578\times10^{-7}$&0.03&12&5&250\\
        \hline
        \multirow{5}{*}{$A$}
           & Type-I & $1.384\times10^{-3}$&0.98 &6.45&3.66&4167\\
          & Type-I integer & $1.391\times10^{-3}$&0.98 &6&4&4453\\
        &  Type-II &$1.36\times10^{-3}$&* & 6.45&3.2 &4193 \\
        &  Type-II integer &$1.367\times10^{-3}$&0.99 & 6&4 &4453 \\
        &Original& $1.182\times10^{-2}$&0.12&12&5&250\\
        \hline    
        \multirow{5}{*}{$V$}
          & Type-I & $2.148\times10^{2}$&0.89 &5.54&27.53&3582\\
          & Type-I integer & $2.158\times10^{2}$&0.89 &6&25&3311\\
        &  Type-II &$1.913\times10^{2}$&* & 5.74&21.4 &3729 \\
        &  Type-II integer &$1.916\times10^{2}$&1 & 6&20 &3583 \\
        &Original& $1.181\times10^{2}$&0.16&12&5&250\\
        \hline
        \end{tabular}
    \label{table:ch2008}
    \end{table}
    
    As shown in Table \ref{table:ch2008}, the optimal type-II design consistently outperforms the optimal type-I design. Notably, the significant improvements under the $D$- and $V$-optimality criteria demonstrate that the choice of inspection interval has a substantial effect on experimental efficiency. In contrast, the $A$-optimality criterion shows only minimal differences between the two designs. When compared to the original experiment, the efficiencies under $D$-, $A$-, and $V$-optimality are 0.14, 0.57, and 0.53, respectively, highlighting that optimal designs can significantly enhance efficiency. Furthermore, the optimal approximate designs are very close to the optimal integer designs, suggesting that the approximate designs can serve as a good starting point for implementing discrete search algorithms.
    We also conduct a sensitivity analysis to demonstrate the robustness of the proposed optimal designs. Since the results are similar and consistently perform well, they are presented in the supplementary material.
\end{example}

\section{Conclusions}
In this study, we conduct a comprehensive and systematic exploration of the optimal design for gamma degradation tests, a topic that has not been thoroughly addressed in the existing literature. Several notable contributions are made. First, we incorporate the minimum inspection interval constraint into the design framework, enhancing its practical applicability in real-world settings. Second, we analytically derive the optimal type-I designs under three scenarios and find, contrary to intuition, that a larger inspection interval does not always lead to improved efficiency. Third, we analytically derive the optimal type-II designs under a cost constraint and show that the most efficient inspection intervals follow the sequence $(T-(m-1)\Delta t, \Delta t, \dots, \Delta t)$ regardless of ordering, while the periodic inspection time is the least efficient.
Theoretical results are supported by numerical examples. Example 1 shows that inspection time planning has little influence on efficiency, whereas Example 2 demonstrates that it can have a substantial impact. A sensitivity analysis demonstrates the robustness of the proposed optimal designs when the estimated parameters deviate from the true parameters. These results offer a reliable and efficient framework for conducting gamma DTs, ensuring accurate data collection to support sound decision-making. Moreover, the two types of optimal designs provide practitioners with the flexibility to adjust the inspection times based on practical constraints or scheduling needs.

The study preliminary focuses on fundamental gamma DTs. In real applications, however, several variants of gamma DTs may arise. For instance, we assume that the gamma degradation process follows a linear degradation path, $E(Z(t)) = e^\gamma t$. In practice, the degradation path may be nonlinear, typically convex or concave. Investigating optimal designs under nonlinear degradation paths would broaden the applicability of our results.
Additionally, this study focuses on the fixed effects model. In contrast, several studies adopt random or mixed effects models to account for unit-to-unit variation among test units \citep{tseng2009optimal, peng2022optimum, cheng2024optimal}. Extending our analytical results to these models is an important direction for future research.
Finally, although many studies have examined ADTs, as summarized in Table~\ref{tab:reference}, a general framework has not yet been proposed. Several theoretical results developed in this work, such as the properties of the polygamma function, may serve as useful tools for extending our framework to the optimal design of gamma ADTs.

% \section*{Acknowledgements}
% This work is support by the National Science and Technology Council of Taiwan (grant no. 111-2118-M-A49-009-MY3).

\section*{Data availability statement}
The data that support the findings of this study are openly available at\\
https://tigerprints.clemson.edu/all\_theses/470, \cite{chaluvadi2008accelerated}.

\if0\blind{
\section*{Acknowledgements}
The authors acknowledge the generous support from the National Science and Technology Council of Taiwan (grant no. 111-2118-M-A49-009-MY3).} \fi

\bibliographystyle{apalike}
\bibliography{refer}
\clearpage
\begin{appendix}
\section{Supplementary material}
\subsection{Supporting Lemmas for Proofs}
\begin{lemma}\citep{yang2017monotonicity}\label{lemma:yang2017}
    Let the functions $A(x)$, $B(x)$ be defined on $(0,\infty)$ such that their Laplace transforms exist with $B(x) \neq 0$ for all $x > 0$. Then the function $\frac{\int_0^\infty A(t)e^{-xt}dt}{\int_0^\infty B(t)e^{-xt}dt}$ is increasing on $(0,\infty)$ if $\frac{A(x)}{B(x)}$ is decreasing on $(0,\infty)$.
\end{lemma}
\begin{lemma}\citep{qi2020some}\label{lemma:qi2020}
    \begin{align*}
        &x\psi_1(x)-1=\int_0^\infty\left(\frac{e^{t}(e^{t}-1-t)}{(e^t-1)^2}\right)e^{-xt}dt>0,\\
        &x^2\psi_1(x)-x-\frac{1}{2}=\int_0^\infty\left(\frac{e^{t}(e^{t}(t-2)+t+2)}{(e^t-1)^3}\right)e^{-xt}dt>0.
    \end{align*}
\end{lemma}
\begin{lemma}\citep{qi2022decreasing}
    The function
    \begin{equation*}
        \frac{2x\psi_1(x)+x^2\psi_2(x)-1}{(x^2\psi_1(x)-x-1/2)^2}
    \end{equation*}
    is a decreasing function in $x\in(0,\infty)$ onto the interval $(-6,-4)$.\label{lemma:qi2022}
\end{lemma}
\begin{lemma}\label{lemma:tung2025}
    The function
    \begin{equation*}
        \Omega(x)=\frac{2x\psi_1(x)+x^2\psi_2(x)-1}{(x\psi_1(x)-1)^2}
    \end{equation*}
    is a decreasing function in $x\in(0,\infty)$ onto the interval $(-2/3,0)$.
\end{lemma}
\begin{proof}
    By differentiating $\Omega(x)$, we obtain
    \begin{align*}
        \Omega(x)'&=\left[\frac{2x\psi_1(x)+x^2\psi_2(x)-1}{(x^2\psi_1(x)-x-1/2)^2}\left(\frac{x^2\psi_1(x)-x-1/2}{x\psi_1(x)-1}\right)^2\right]'\\
        &=\left[\frac{2x\psi_1(x)+x^2\psi_2(x)-1}{(x^2\psi_1(x)-x-1/2)^2}\right]'\left(\frac{x^2\psi_1(x)-x-1/2}{x\psi_1(x)-1}\right)^2\\
        &+\frac{2x\psi_1(x)+x^2\psi_2(x)-1}{(x^2\psi_1(x)-x-1/2)^2}\left[\left(\frac{x^2\psi_1(x)-x-1/2}{x\psi_1(x)-1}\right)^2\right]'.
    \end{align*}
    According to lemmas \ref{lemma:yang2017}, \ref{lemma:qi2020} and \ref{lemma:qi2022}, if
    \begin{equation}\label{eq:omegad}
        \left(\frac{e^{t}(e^{t}(t-2)+t+2)}{(e^t-1)^3} \middle/\frac{e^{t}(e^{t}-1-t)}{(e^t-1)^2}\right)'<0,
    \end{equation}
    then $$\left(\frac{x^2\psi_1(x)-x-1/2}{x\psi_1(x)-1}\right)'>0,$$
    implying $\Omega(x)'<0$.
    Computing equation (\ref{eq:omegad}) gives
    % \begin{align*}
    %     &\frac{-(t-3)e^{3t}-(2t+7)e^{2t}+(2t^2+3t+5)e^t-1}{(e^t-1)^2(e^t-1-t)^2}\\
    %     &=-\frac{1}{(e^t-1)^2(e^t-1-t)^2}\\
    %     &\left(\frac{1}{36}t^6 +\frac{2}{45}t^7 +\frac{3}{80}t^8 +\frac{167}{7560}t^9+\sum\limits_{k=9}^\infty\left(3^{k-1}(k-9)+k(2^k-2k)+7\times2^k-k-5\right)\frac{t^k}{k!}\right)<0.
    % \end{align*}
    \begin{align*}
        &\frac{-(t-3)e^{3t}-(2t+7)e^{2t}+(2t^2+3t+5)e^t-1}{(e^t-1)^2(e^t-1-t)^2}\\
        &=\frac{-1}{(e^t-1)^2(e^t-1-t)^2}\left(
        \begin{aligned}
            &\frac{1}{36}t^6 +\frac{2}{45}t^7 +\frac{3}{80}t^8 +\frac{167}{7560}t^9\\
            &+\sum\limits_{k=9}^\infty\left(3^{k-1}(k-9)+k(2^k-2k)+7\times2^k-k-5\right)\frac{t^k}{k!}
        \end{aligned}\right)<0.
    \end{align*}
    Hence, $\Omega(x)$ is a decreasing function. Next, using the asymptotic approximation,
    \begin{equation*}
        (-1)^n\psi_n(x)\sim \frac{(n-1)!}{x^n}+\frac{(n)!}{2x^{n+1}}+\frac{(n+1)!}{12x^{n+2}}+\mathcal{O}\left(\frac{1}{x^{n+3}}\right),\text{ as }x\rightarrow\infty,
    \end{equation*}
    we have $\lim_{x\rightarrow\infty}\Omega(x)=-2/3$.
    In addition, since
    \begin{equation*}
        (-1)^n\psi_n(x)=n!\sum\limits_{v=0}^\infty\frac{1}{(x+v)^{n+1}},
    \end{equation*}
    it follows that $\lim_{x\rightarrow 0^{+}}\Omega(x)=0$. 
    Consequently, $\Omega(x)$ is a decreasing function in $x\in(0,\infty)$ onto the interval $(-2/3,0)$.
\end{proof}
\begin{lemma} \label{lemma:alzer20012004}
    For any $y>x>0$, we have
    \begin{equation*}
        x\psi_1(x)-2y\psi_1(y)-y^2\psi_2(y)>0.
    \end{equation*}
\end{lemma}
\begin{proof}
    According to Lemma 2.3 in \cite{alzer2004sharp}, we have
    \begin{equation*}
        x\psi_1(x)-y\psi_1(y)>0,\text{ for }y>x>0.
    \end{equation*}
    According to Corollary 1 in \cite{alzer2001mean}, we have 
    \begin{equation*}
        -y\psi_1(y)-y^2\psi_2(y)>0,\text{ for }y>0.
    \end{equation*}
    Hence, 
    \begin{equation*}
        x\psi(x)-2y\psi_1(y)-y^2\psi_2(y)>0,\text{ for }y>x>0.
    \end{equation*}
\end{proof}

\subsection{The proof of Theorem \ref{thm:EI exact optimal design}}\label{appendix:proof thm EI exact optimal design}
Since the proofs are similar, we only present the case of D-optimality.
The optimization problem can be viewed as
\begin{equation*}
\begin{aligned}
    \max  (nm)^2 (\alpha\tau^3\psi_1(\alpha\tau)-\tau^2)
\end{aligned}
\end{equation*}
subject to
\begin{equation*}
\begin{aligned}
    &C_{it} n+C_{mea} mn+ C_{op} m\tau = 1,\\
    &\tau \geq \Delta t,\\
    &n,m \geq 1.    
\end{aligned}
\end{equation*}
For notational simplicity, we define
\begin{equation*}
    \varrho(\tau) = \alpha \tau^3 \psi_1(\alpha \tau) - \tau^2,
\end{equation*}
and hence,
\begin{equation*}
    \varphi(\tau) = \frac{1}{2} \frac{d }{d\tau}\log \varrho(\tau).
\end{equation*}
Based on KKT conditions,
\begin{equation*}   
     L(\zeta,\lambda,\mu_1,\mu_2)=n^2m^2 \varrho(\tau)-\lambda(C_{it} n+C_{mea} m n+ C_{op} m\tau - 1)+\mu_1(n-1)+\mu_2(m-1)+\mu_3(\tau-\Delta t).
\end{equation*} 
\begin{subnumcases}{}
    \pdv{L}{n}=2nm^2 \varrho(\tau)-\lambda (C_{it}+C_{mea}m)+\mu_1=0 \label{DKKT1}\\
    \pdv{L}{m}=2n^2m \varrho(\tau)-\lambda(C_{mea}n+C_{op}\tau)+\mu_2=0 \label{DKKT2}\\
    \pdv{L}{t}=n^2m^2 \varrho^{'}(\tau)-\lambda(C_{op}m)+\mu_3=0 \label{DKKT3}\\
    \pdv{L}{\lambda}=C_{it} n+C_{mea} m n+ C_{op} m\tau- 1=0 \label{DKKT4}\\
    \mu_1(n-1)=0 \nonumber\\
    \mu_2(m-1)=0 \nonumber\\
    \mu_3(\tau-\Delta t)=0 \nonumber\\
    \mu_1 \geq 0 \nonumber\\
    \mu_2 \geq 0 \nonumber\\
    \mu_3 \geq 0 \nonumber\\
    n \geq 1 \nonumber\\
    m \geq 1 \nonumber\\
    \tau\geq \Delta t \nonumber
\end{subnumcases}

\begin{enumerate}    
    \item For, $\mu_1 >0$, $\mu_2 =0$ and $\mu_3=0$, we have $n^*=1$, $m>1$ and $\frac{1-C_{it}-C_{mea}}{C_{op}}>\tau>\Delta t$.
     From (\ref{DKKT2}), (\ref{DKKT3}) and (\ref{DKKT4}), we have
     \begin{gather*}
         -2m\varrho(\tau)+\lambda(C_{mea}+C_{op}\tau)=0,\\
         \lambda=\frac{\varrho^{'}(\tau)m}{C_{op}},\\
         m(\tau)=\frac{1-C_{it}}{C_{mea}+C_{op}\tau}>1.
     \end{gather*}
    Based on some algebraic manipulations, $\tau_1$ is the solution of
    \begin{equation}
        \varphi(\tau)=K_1(\tau),
        \label{D case 2 t}
    \end{equation}
    and
    \begin{equation}
        m_1=m(\tau_1).
        \label{D case 2 m}
    \end{equation} \\
    From (\ref{DKKT1}), (\ref{D case 2 m}) and (\ref{D case 2 t}), we have
    \begin{align*}
        &\mu_1=-2m_1\varrho(\tau_1)+\frac{\varrho^{'}(\tau_1)}{C_{op}}(C_{it}+C_{mea}m_1)>0,\\    
        \Leftrightarrow\quad&\frac{\varrho^{'}(\tau_1)}{\varrho(\tau_1)}>\frac{2C_{op}m_1}{C_{it}+C_{mea}m_1},\\    
        \Leftrightarrow\quad&\frac{C_{op}}{C_{mea}+C_{op}\tau_1}>\frac{C_{op}m_1}{C_{it}+C_{mea}m_1},\\
        \Leftrightarrow\quad& C_{it}C_{mea}>C_{op}(1-2C_{it})\tau_1.
    \end{align*}
    Note that $1-2C_{it}-C_{mea}\geq 0 \Leftrightarrow \frac{1-C_{it}-C_{mea}}{C_{op}}\geq\frac{C_{it}C_{mea}}{C_{op}(1-2C_{it})}$. Therefore, if $1-2C_{it}-C_{mea} \geq 0$ and $\frac{C_{it}C_{mea}}{C_{op}(1-2C_{it})}>\tau_1>\Delta t$, or $1-2C_{it}-C_{mea} < 0$ and $\frac{1-C_{it}-C_{mea}}{C_{op}}>\tau_1>\Delta t$, then $n^*=1$, $m^*=\frac{1-C_{it}}{C_{mea}+C_{op}\tau_1}$ and $\tau^*=\tau_1$.

    \item For $\mu_1 =0$, $\mu_2 >0$ and $\mu_3=0$, we have $n>1$, $m^*=1$ and $\frac{1-C_{it}-C_{mea}}{C_{op}}>\tau>\Delta t$.
     From (\ref{DKKT1}), (\ref{DKKT3}) and (\ref{DKKT4}), we have
     \begin{gather*}
         -2n\varrho(\tau)+\lambda(C_{it}+C_{mea})=0,\\
         -n^2\varrho^{'}(\tau)+\lambda C_{op}=0,\\
         n(\tau)=\frac{1-C_{op}\tau}{C_{it}+C_{mea}}.
     \end{gather*}
    Therefore, $\tau_2$ is the solution of
    \begin{equation}
        \varphi(\tau)=K_2(\tau),
        \label{D case 3 t}
    \end{equation}
    and
    \begin{equation}
        n_2=n(\tau_2)>1.
        \label{D case 3 n}
    \end{equation} \\
   From (\ref{DKKT2}), (\ref{D case 3 t}) and (\ref{D case 3 n}), we have,
    \begin{align*}
        &\mu_2=-2n_2\varrho(\tau_2)+\frac{n_2\varrho^{'}(\tau_2)}{C_{op}}(C_{mea}n_2+C_{op}\tau_2)>0,\\   
        \Leftrightarrow\quad&\frac{\varrho^{'}(\tau_2)}{\varrho(\tau_2)}>\frac{2C_{op}}{C_{mea}n_2+C_{op}\tau_2},\\
        \Leftrightarrow\quad&\frac{C_{op}}{n_2(C_{it}+C_{mea})}>\frac{C_{op}}{C_{mea}n_2+C_{op}\tau_2},\\
        \Leftrightarrow\quad&\tau_2>\frac{C_{it}}{C_{op}(C_{mea}+2C_{it})}.
    \end{align*}
    Note that $1-2C_{it}-C_{mea}\geq 0 \Leftrightarrow \frac{1-C_{it}-C_{mea}}{C_{op}}\geq\frac{C_{it}}{C_{op}(C_{mea}+2C_{it})}$. Therefore, if $1-2C_{it}-C_{mea}> 0$ and $\frac{1-C_{it}-C_{mea}}{C_{op}}>\tau_2>\max\left(\Delta t,\frac{C_{it}}{C_{op}(C_{mea}+2C_{it})}\right)$, then $n^*=\frac{1-C_{op}\tau_2}{C_{it}+C_{mea}}$, $m^*=1$ and $\tau^*=\tau_2$.
    
    \item For $\mu_1 =0$, $\mu_2 =0$ and $\mu_3=0$, we have $n>1$, $m>1$ and $\frac{1-C_{it}-C_{mea}}{C_{op}}>\tau>\Delta t$.
     From (\ref{DKKT1}), (\ref{DKKT2}), (\ref{DKKT3}) and \ref{DKKT4},
     \begin{gather*}
         -2nm^2\varrho(\tau)+\lambda(C_{it}+C_{mea}m)=0,\\
         -2n^2m\varrho(\tau)+\lambda(C_{mea}n+C_{op}\tau)=0,\\
         -n^2m^2\varrho^{'}(\tau)+\lambda C_{op}m=0,\\
         C_{it}n+C_{mea}mn+C_{op}m\tau-1=0.
     \end{gather*}
    Based on some algebraic manipulations, $\tau_3$ is the solution of
    \begin{equation*}
        \varphi(\tau)=K_3(\tau),
        \label{D case 4 t}
    \end{equation*}
    and
    \begin{gather*}
        n(\tau)=\frac{-C_{it}+\sqrt{C_{it}^2+\frac{C_{mea}C_{it}}{C_{op}\tau}}}{\frac{C_{mea}C_{it}}{C_{op}\tau}}>1,\\
        m(\tau)=\frac{-C_{it}+\sqrt{C_{it}^2+\frac{C_{mea}C_{it}}{C_{op}\tau}}}{C_{mea}}>1.
    \end{gather*}
    Note that, $1-2C_{it}-C_{mea}\geq 0 \Leftrightarrow \tau_u\geq \tau_l$, where $\tau_u=\frac{C_{it}}{C_{op}(C_{mea}+2C_{it})}$ and $\tau_l=\frac{C_{mea}C_{it}}{C_{op}(1-2C_{it})}$, and $n(\tau)>1$ and $m(\tau)>1 \Leftrightarrow \tau_l< \tau< \tau_u$. Therefore, if $1-2C_{it}-C_{mea}\geq 0$ and $\max(\Delta t,\tau_l)<\tau_3 < \tau_u$, then $n^*=n(\tau_3)$, $m^*=m(\tau_3)$ and $\tau^*=\tau_3$.

    \item For $\mu_1 >0$, $\mu_2 >0$, and $\mu_3=0$, we have $n^*=m^*=1$, $\tau^*=\frac{1-C_{mea}-C_{it}}{C_{op}}$.
    From (\ref{DKKT1}), (\ref{DKKT2}) and (\ref{DKKT3}), we have,
    \begin{equation*}
    \begin{aligned}
        &\mu_1=-2\varrho(\tau^*)+\frac{\varrho^{'}(\tau^*)}{C_{op}}(C_{it}+C_{mea})>0,\\
        &\mu_2=-2\varrho(\tau^*)+\frac{\varrho^{'}(\tau^*)}{C_{op}}(1-C_{it})>0.
    \end{aligned}
    \end{equation*}
    Hence, if $\varphi(\tau^*)>\max \bigg( \frac{C_{op}}{C_{it}+C_{mea}},\frac{C_{op}}{1-C_{it}}\bigg)$ then $n^*=m^*=1,\tau^*=\frac{1-C_{mea}-C_{it}}{C_{op}}$.
    
    \item For, $\mu_1 >0$, $\mu_2 =0$ and $\mu_3>0$, we have $n^*=1$, $m>1$ and $\tau^*=\Delta t$.
     From (\ref{DKKT4}), (\ref{DKKT2}), (\ref{DKKT1}) and (\ref{DKKT3}), we have
     \begin{gather*}
         m=\frac{1-C_{it}}{C_{mea}+C_{op}\Delta t}>1,\\
         \lambda=\frac{2m\varrho(\Delta t)}{C_{mea}+C_{op}\Delta t},\\         
         \mu_1=-2m^2\varrho(\Delta t)+\lambda (C_{it}+C_{mea}m)>0,\\
         \mu_3=-m^2\varrho'(\Delta t)+\lambda (C_{op}m)>0
     \end{gather*}
    Therefore, if $1-2C_{it}-C_{mea}\geq 0$ and $\frac{C_{it}C_{mea}}{C_{op}(1-2C_{it})}>\Delta t$ hold, or $1-2C_{it}-C_{mea}<0$ holds, along with $\varphi(\Delta t)<K_1(\Delta t)$, then $n^*=1$, $m^*=\frac{1-C_{it}}{C_{mea}+C_{op}\Delta t}$ and $\tau^*=\Delta t$.

    \item For $\mu_1 =0$, $\mu_2 >0$ and $\mu_3>0$, we have $n>1$, $m^*=1$ and $\tau^*=\Delta t$.
     From (\ref{DKKT4}), (\ref{DKKT1}), (\ref{DKKT2}) and (\ref{DKKT3}), we have
     \begin{gather*}
         n^*=\frac{1-C_{op}\Delta t}{C_{it}+C_{mea}}>1,\\
         \lambda=\frac{2n\varrho(\Delta t)}{C_{it}+C_{mea}},\\
         \mu_2=-2n^2\varrho(\Delta t)+\lambda (C_{mea}n+C_{op}\Delta t)>0,\\
         \mu_3=-n^2\varrho'(\Delta t)+\lambda C_{op}>0.
     \end{gather*}
    Therefore, if $\Delta t>\frac{C_{it}}{C_{op}(C_{mea}+2C_{it})}$ and $\varphi(\Delta t)<K_2(\Delta t)$, then $n^*=\frac{1-C_{op}\Delta t}{C_{it}+C_{mea}}$, $m^*=1$ and $\tau^*=\Delta t$.
    
    \item For $\mu_1 =0$, $\mu_2 =0$ and $\mu_3>0$, we have $n>1$, $m>1$ and $\tau^*=\Delta t$.
     From (\ref{DKKT1}), (\ref{DKKT2}), (\ref{DKKT4}) and \ref{DKKT3}, we have
     \begin{gather*}
         \lambda=\frac{2nm^2\varrho(\Delta t)}{C_{it}+C_{mea}m}=\frac{2n^2m\varrho(\Delta t)}{C_{mea}n+C_{op}\Delta t},\\
         C_{it}n+C_{mea}mn+C_{op}m\Delta t-1=0,\\
         \mu_3=-n^2m^2\varrho'(\Delta t)+\lambda (C_{op}m)>0.
     \end{gather*}
    Based on some algebraic manipulations,
    \begin{gather}
        n^*=\frac{-C_{it}+\sqrt{C_{it}^2+\frac{C_{mea}C_{it}}{C_{op}\Delta t}}}{\frac{C_{mea}C_{it}}{C_{op}\Delta t}}>1,\label{D case 8 n}\\ 
        m^*=\frac{-C_{it}+\sqrt{C_{it}^2+\frac{C_{mea}C_{it}}{C_{op}\Delta t}}}{C_{mea}}>1,\label{D case 8 m}\\ 
        \frac{\varrho'(\Delta t)}{\varrho(\Delta t)}<\frac{2C_{op}}{C_{mea}n^*+C_{op}\Delta t}.\nonumber
    \end{gather}
    In addition, $1-2C_{it}-C_{mea}>0$, $n^*>1$ and $m^*>1$ imply $\tau_l=\frac{C_{mea}C_{it}}{C_{op}(1-2C_{it})}<\Delta t<\tau_u=\frac{C_{it}}{C_{op}(C_{mea}+2C_{it})}$.
    Therefore, if $\varphi(\Delta t)<K_3(\Delta t)$, $1-2C_{it}-C_{mea}>0$ and $\tau_l<\Delta t < \tau_u$, then $n^*$ and $m^*$ are correspond to equations (\ref{D case 8 n}) and (\ref{D case 8 m}), and $\tau^*=\Delta t$.

    \item For $\mu_1 >0$, $\mu_2 >0$, and $\mu_3>0$, we have $n^*=m^*=1$, $\tau^*=\Delta t$. This case only happens when $C_{it}+C_{mea}+C_{op}\Delta t=1$.    
    \end{enumerate}
\subsection{Sensitivity analysis}
As shown in Sections 3 and 4, the optimal design is derived assuming the parameters are known in advance. These parameters are typically estimated from phase-I information, such as historical data or a pilot study. Consequently, uncertainty in parameter estimation from phase-I may affect the resulting optimal design. To evaluate this impact, we conduct a simulation to assess the efficiency of the optimal design when the estimated parameters deviate from the true parameters.

We use the setting from Example 2, assuming the true parameters are $\alpha = 0.028$ and $\gamma = -2.073$. The estimated parameters are assumed to deviate from the true values by $0$, $\pm1$, $\pm2$, and $\pm3$ times the estimated standard deviations, where $\sigma_\alpha = \sqrt{\text{Var}(\hat{\alpha})} = 4.67 \times 10^{-3}$ and $\sigma_\gamma = \sqrt{\text{Var}(\hat{\gamma})} = 1.09 \times 10^{-1}$. 
Since the $D$- and $A$-optimality criteria do not depend on $\gamma$, this results in 7 parameter combinations for evaluating $D$- and $A$-optimality, and 49 combinations for $V$-optimality.
For each combination, we first determine the suboptimal design based on the estimated parameters. The objective function is then evaluated under the true parameters, and the resulting value is compared to the value of the optimal design under the true parameters to assess relative efficiency.

Tables \ref{table:senI} and \ref{table:senII} summarize the REs under various levels of parameter deviation for type-I and type-II designs. The results show that the REs of the proposed optimal type-I and type-II designs all exceed $95\%$, indicating strong robustness for our proposed method.

\begin{table}[H]
\caption{Relative efficiency of type-I designs under various levels of parameter deviation}
\centering
    \begin{tabular}[c]{|c|c|c c c c c c c|}
    \hline
    \multirow{2}{*}{Criterion} & \multirow{2}{*}{$\gamma$} & \multicolumn{7}{c|}{$\alpha$}\\
    \cline{3-9} &  & $-3\sigma_\alpha$ & $-2\sigma_\alpha$ & $-\sigma_\alpha$ & – & $+\sigma_\alpha$ & $+2\sigma_\alpha$ &$+3\sigma_\alpha$ \\
    \hline
    $D$ & – & 99.95\% & 99.99\% & 100\% &100\% & 100\%&99.99\%&99.97\% \\
    \hline
    $A$ & – & 94.85\% &98.27\%	& 99.65\% & 100\% & 99.75\%&99.13\%&98.26\% \\
    \hline
    \multirow{7}{*}{$V$} & $-3\sigma_\gamma$ & 99.16\% & 99.48\% & 99.72\% &99.88\% & 99.97\% & 100\% &99.98\% \\
    & $-2\sigma_\gamma$ & 99.37\% & 99.66\% & 99.86\% &99.97\% & 100\% & 99.95\% &99.83\% \\
    & $-\sigma_\gamma$ & 99.47\% & 99.75\% & 99.92\% &100\% & 99.98\% & 99.87\% &99.68\% \\
    & – & 99.53\% & 99.79\% & 99.95\% &100\% & 99.95\% & 99.81\% &99.58\% \\
    & $+\sigma_\gamma$ & 99.57\% & 99.82\% & 99.97\% &100\% & 99.93\% & 99.76\% &99.49\% \\
    & $+2\sigma_\gamma$ & 99.60\% & 99.84\% & 99.97\% &100\% & 99.91\% & 99.72\% &99.43\% \\
    & $+3\sigma_\gamma$ & 99.62\% & 99.86\% & 99.98\% &100\% & 99.89\% & 99.68\% &99.37\% \\
    \hline    
    \end{tabular}
\label{table:senI}
\end{table}

\begin{table}[H]
\caption{Relative efficiency of type-II designs under various levels of parameter deviation}
\centering
    \begin{tabular}[c]{|c|c|c c c c c c c|}
    \hline
    \multirow{2}{*}{Criterion} & \multirow{2}{*}{$\gamma$} & \multicolumn{7}{c|}{$\alpha$}\\
    \cline{3-9} &  & $-3\sigma_\alpha$ & $-2\sigma_\alpha$ & $-\sigma_\alpha$ & – & $+\sigma_\alpha$ & $+2\sigma_\alpha$ &$+3\sigma_\alpha$ \\
    \hline
    $D$ & – & 100\% & 100\% & 100\% &100\% & 100\%&100\%&100\% \\
    \hline
    $A$ & – & 96.01\% &98.64\%	& 99.73\% & 100\% & 99.80\%&99.29\%&98.57\% \\
    \hline
    \multirow{7}{*}{$V$} & $-3\sigma_\gamma$ & 99.12\% & 99.43\% & 99.67\% &99.83\% & 99.94\% & 99.99\% &100\% \\
    & $-2\sigma_\gamma$ & 99.37\% & 99.64\% & 99.84\% &99.95\% & 100\% & 99.98\% &99.91\% \\
    & $-\sigma_\gamma$ & 99.51\% & 99.76\% & 99.92\% &99.99\% & 99.99\% & 99.92\% &99.78\% \\
    & – & 99.61\% & 99.83\% & 99.96\% &100\% & 99.96\% & 99.84\% &99.65\% \\
    & $+\sigma_\gamma$ & 99.69\% & 99.88\% & 99.98\% &99.99\% & 99.92\% & 99.77\% &99.54\% \\
    & $+2\sigma_\gamma$ & 99.75\% & 99.92\% & 99.99\% &99.98\% & 99.88\% & 99.69\% &99.43\% \\
    & $+3\sigma_\gamma$ & 99.79\% & 99.94\% & 100\% &99.96\% & 99.84\% & 99.62\% &99.33\% \\
    \hline    
    \end{tabular}
\label{table:senII}
\end{table}
\end{appendix}
\end{document}